\documentclass[aps,prl,reprint,nobalancelastpage,nofootinbib]{revtex4-2}
\pdfoutput=1

\usepackage[utf8]{inputenc}

\usepackage{amsmath}
\usepackage{amssymb}
\usepackage[all]{xy}
\usepackage{tikz}
\usetikzlibrary{calc}

\usepackage[toc,page]{appendix}
\usepackage{enumerate}
\usepackage{tensor}
\usepackage{booktabs}
\usepackage{bbm}
\usepackage{youngtab}
\usepackage{slashed}
\usepackage{multirow}
\usepackage{makecell}
\usepackage{float}
\usepackage{comment}
\usepackage{verbatim}
\usepackage{xcolor}
\usepackage[framemethod=tikz]{mdframed}
\usepackage{soul}
\usepackage{mathtools}
\usepackage[inline]{enumitem}
\usepackage{alphabeta}
\usepackage{accents}

\usepackage{amsthm}
\newtheorem{theorem}{Theorem}
\newtheorem{lemma}{Lemma}
\newtheorem{corollary}{Corollary}[lemma]
\newtheorem{remark}{Remark}[lemma]

\newcommand{\ab}{|}
\newcommand{\der}{\partial}
\newcommand{\de}{\mathrm{d}}

\newcommand{\x}{\mathrm{x}}
\newcommand{\w}{\mathrm{w}}
\newcommand{\y}{\mathrm{y}}

\newcommand{\rmxi}{\text{\xi}}
\newcommand{\rmvarphi}{\text{\varphi}}
\newcommand{\bx}{\overline{\x}}
\newcommand{\bw}{\overline{\w}}
\newcommand{\by}{\overline{\y}}

\newcommand{\hx}{\hat{\x}}
\newcommand{\hw}{\hat{\w}}
\newcommand{\hy}{\hat{\y}}

\newcommand{\sdil}{\phi}

\newcommand{\tsdil}{\tilde{\phi}}
\newcommand{\tomega}{\tilde{\omega}}
\newcommand{\tdelta}{\tilde{\delta}}

\newcommand{\tsigma}{\tilde{\sigma}}

\newcommand{\e}{\mathrm{e}}
\newcommand{\I}{\mathrm{i}}

\newcommand{\NS}{\mathrm{NS}}
\newcommand{\R}{\mathrm{R}}

\allowdisplaybreaks

\usepackage{hyperref}
\hypersetup{
    colorlinks=true,
    linkcolor=magenta,
    citecolor=magenta,
    filecolor=magenta,      
    urlcolor=magenta,
}

\begin{document}
\numberwithin{equation}{section}

\title{Analytic bounds on late-time axion-scalar cosmologies}
\author{Gary Shiu}
\email{shiu@physics.wisc.edu}
\affiliation{Department of Physics, University of Wisconsin-Madison, 1150 University Avenue, Madison, WI 53706, USA}
\author{Flavio Tonioni}
\email{flavio.tonioni@kuleuven.be}
\affiliation{Instituut voor Theoretische Fysica, KU Leuven, Celestijnenlaan 200D, B-3001 Leuven, Belgium}
\author{Hung V. Tran}
\email{hung@math.wisc.edu}
\affiliation{Department of Mathematics, University of Wisconsin-Madison, 480
Lincoln Drive, Madison, WI 53706, USA}

\begin{abstract}
The cosmological dynamics of multiple scalar/pseudoscalar fields are difficult to solve, especially when the field-space metric is curved. This presents a challenge in determining whether a given model can support cosmic acceleration, without solving for the on-shell solution. In this work, we present bounds on late-time FLRW-cosmologies in classes of theories that involve arbitrary numbers of scalar and pseudoscalar fields coupled both kinetically (leading to a curved field space metric) and through scalar potentials. Such bounds are proven analytically, independently of initial conditions, with no approximation in the field equations and without referring to explicit solutions. Besides their broad applications to cosmological model building, our bounds can be applied to studying asymptotic cosmologies of certain classes of string compactifications.
\end{abstract}

\maketitle

\section{Introduction}

Since the first direct experimental evidence for dark energy, the origin of cosmic acceleration from a fundamental theory of gravity has been a riveting puzzle.
The simplest source of dark energy is a positive cosmological constant.
Yet, cosmic acceleration only requires the parameter $w$ of the equation of state of dark energy to be $w<-1/3$, or equivalently $\epsilon \equiv - \dot{H}/{H^2} < 1$, where $H$ is the Hubble parameter.
Except for the case of a cosmological constant ($w=-1$), dark energy evolves with time -- a scenario often known as ``quintessence'' \cite{Ratra:1987rm,Wetterich:1987fm,Caldwell:1997ii}.
Recently, the first-year observations of the Dark Energy Spectroscopic Instrument (DESI) give a tantalizing hint of an evolving dark-energy equation of state \cite{DESI:2024mwx}.
While it is too early to tell with the current statistical significance whether the $\Lambda$CDM-model is disfavored, there is enough of a motivation to revisit the possibility of a varying dark energy; see e.g. refs.~\cite{Tada:2024znt, Payeur:2024kyy, Berghaus:2024kra, Andriot:2024jsh, Bhattacharya:2024hep, Ramadan:2024kmn} for recent works on quantifying the DESI results in terms of the dark-energy potential.
In fact, prior to this recent experimental development, a number of independent works, inspired in part by Swampland considerations \cite{Obied:2018sgi, Ooguri:2018wrx}, studied accelerating cosmologies with rolling scalar fields; see e.g. refs. \cite{Agrawal:2018own, Agrawal:2018rcg, Hebecker:2019csg, Cicoli:2020cfj, Rudelius:2021azq, Cicoli:2021fsd, Rudelius:2022gbz, Andriot:2022xjh, Marconnet:2022fmx, Shiu:2023nph, Shiu:2023fhb, VanRiet:2023cca, Andriot:2023wvg}.

In particular, in refs.~\cite{Shiu:2023nph, Shiu:2023fhb}, the present authors derived theoretical late-time bounds on a class of quintessence models.
For flat $d$-dimensional FLRW-cosmologies, we showed that the late-time $\epsilon$-parameter of a theory of canonical scalars coupled through multi-exponential potentials is bounded from below as $\smash{\epsilon \geq [(d-2)/4] \, (\gamma_\infty)^2}$, where $\gamma_\infty$ is the minimal-length vector joining the coupling convex hull to the origin \cite{Shiu:2023nph}.
For positive-definite potentials, and up to a few assumptions on the shape of the coupling convex hull, we then showed that there is a unique late-time attractor solution, which furthermore saturates the bound \cite{Shiu:2023fhb}.
Such analyses provided non-perturbative proofs of the late-time behavior of the cosmological solutions, going beyond the linear-stability analyses dating back to refs.~\cite{Copeland:1997et, Collinucci:2004iw, Hartong:2006rt}.\footnote{Throughout this article, by ``non-perturbative'' proofs and convergence, we mean that we describe any possible general solution of the classical equations of motion; we do not merely perform perturbative linear stability analyses of exact solutions. Of course, for string-theoretic applications, the classical equations we consider here originate perturbatively in the string-coupling and $\alpha'$-expansions.}
Despite this progress, there are some important extensions that would enhance the reach of our approach.
This is because the scalars in string theory (and supergravity) that have an unlimited field range typically come with their axionic superpartners.

The bounds in refs.~\cite{Shiu:2023fhb,Shiu:2023nph} were derived under the assumption that all the axions that are pseudoscalars with a compact field space are stabilized.
This can be achieved in string compactifications, yet one should also classify cases where some of the axions remain dynamical.
If the axions partake in the dynamics, the cosmological equations governing the multi-field system are generally formidably difficult to analyze due to curvature in the field space, which implies kinetic couplings.
The aim of the present work is to cover a wider class of quintessence models that can arise from a fundamental theory of gravity, extending our earlier analysis to theories with multiple scalars and axions, in which the field space metric is negatively curved.
There exist large classes of string compactifications where this is the case in asymptotic limits \cite{Ooguri:2006in}, although exceptions exist \cite{trenner2010asymptotic, Marchesano:2023thx, Raman:2024fcv}.
We believe that our current study is a necessary step towards a fully complete characterization of the late-time dynamics governing multi-field cosmologies in string theory.

We will present general classes of theories with fields coupled exponentially both in the kinetic and potential energy terms in which the late-time bounds for canonical scalars are unaltered, providing geometric and physical arguments in the main text, substantiated by rigorous mathematical proofs in app. \ref{app: late-time scalar-axion cosmologies}.
In the derivation of all our conclusions, we do not make any approximation on the field equations, we do not need to find an explicit solution to such equations, and we do not make reference to the initial conditions: the bounds apply to all possible fully-fledged time-dependent solutions of the cosmological equations.
Across the literature, this is unique to our approach \cite{Shiu:2023nph, Shiu:2023fhb}.
Studies of analogous theories based on the linear stability of exact solutions and numerical checks appear e.g. in refs.~\cite{Sonner:2006yn, Cicoli:2020cfj, Cicoli:2020noz, Russo:2022pgo, Brinkmann:2022oxy, Revello:2023hro, Seo:2024qzf}.\footnote{See also refs.~\cite{Copeland:1997et, Collinucci:2004iw, Hartong:2006rt}.
For further recent studies, see e.g. refs.~\cite{Conlon:2022pnx, Rudelius:2022gbz, Marconnet:2022fmx, Apers:2022cyl, Hebecker:2023qke, VanRiet:2023cca, Andriot:2023wvg, Apers:2024ffe}.}
As far as our current investigation is concerned, these earlier studies provide valuable examples to test our bounds.
Another aspect in which our results go beyond the existing literature \cite{Sonner:2006yn, Cicoli:2020cfj, Cicoli:2020noz, Russo:2022pgo, Brinkmann:2022oxy, Revello:2023hro, Seo:2024qzf} lies in the fact that, for the cases we constrain, they apply to non-diagonal kinetic couplings and to arbitrary numbers of potential terms.
Finally, we stress that our methods are also applicable to the study of cosmic contraction, as shown in ref.~\cite{Shiu:2023yzt}.
Models involving exponential couplings in kinetic and negative potential terms appear e.g. in refs.~\cite{Li:2013hga, Fertig:2013kwa, Li:2014qwa, Levy:2015awa, Ijjas:2021zyf, Quintin:2024boj}.

Although our results are independent of the stability of critical-point solutions, we discuss the critical points (known and new ones) as well, and speculate on which one might presumably be the late-time attractor.
A fully general list of the critical points is in app. \ref{app: late-time scalar-axion cosmologies}.
In app. \ref{app.: axion couplings in string-theoretic models}, we report the scalings of the kinetic couplings of RR- and NSNS-axions to the dilaton and the universal radion.
This can serve as a testing ground for our results in string compactifications; for complex-structure moduli, more sophisticated tools are needed and we refer to the literature; see e.g. refs.~\cite{Grimm:2004uq, Grimm:2004ua, Grimm:2019ixq}.

\section{Bounds on cosmic acceleration} \label{sec.: bounds cosmic acceleration}
In this section, we present and discuss three general scenarios for axion-scalar cosmologies in which, at arbitrarily late times, we can bound analytically the $\epsilon$-parameter.
Before going into the details, we contextualize the setup we consider and provide definitions that we use to formulate our results.

Let the non-compact spacetime be described by the FLRW-metric
\begin{equation} \label{FLRW-metric}
    d \tilde{s}_{1,d-1}^2 = - \de t^2 + a^2(t) \, \de l_{\mathbb{R}^{d-1}}^2.
\end{equation}
Here, $a$ is the scale factor, which defines the Hubble parameter $\smash{H = \dot{a}/a}$, and $\mathbb{R}^{d-1}$ is the flat $(d-1)$-dimensional Euclidean space.
In this article, we study cosmologies involving $n$ canonical scalars $\phi^a$ and $p$ axions $\zeta^r$, for $a=1,\dots,n$ and $r=1,\dots,p$.\footnote{This terminology is motivated by the fact that string-theoretic axions typically feature kinetic couplings of this kind. There is no additional meaning to the word ``axions'' in this article.}
Assuming for simplicity that all variables only depend on cosmological time, we consider the total kinetic terms of the form
\begin{equation} \label{axion-scalar kinetic energy}
    T[\phi, \zeta] = \dfrac{1}{2} \sum_{a=1}^n (\dot{\phi^a})^2 + \dfrac{1}{2} \sum_{r=1}^p \e^{-\kappa_d \sum_{a} \lambda_{r a} \phi^a} \, (\dot{\zeta}^r)^2,
\end{equation}
i.e. the scalars are kinetically-coupled to the axions as parameterized by the couplings $\lambda_{r a}$.
Here, we consider multi-field multi-exponential scalar potentials of the form
\begin{equation} \label{generic exponential potential}
    V[\phi] = \sum_{i = 1}^m \Lambda_i \, \e^{- \kappa_d \sum_a \gamma_{i a} \phi^a},
\end{equation}
where the $\Lambda_i$-terms are positive constants. Cosmologies with a (pseudo)scalar sector of the form in eqs.~(\ref{axion-scalar kinetic energy}, \ref{generic exponential potential}) are of interest to cosmological model building.\footnote{For the roles of axions/saxions in constructing particle physics models from string theory, see the recent review ref.~\cite{Marchesano:2024gul}.}
Furthermore, there exist classes of string compactifications whose effective field theories take this form.
In fact, it was argued that, in the asymptotic regions of the field space, this form of the potential is indeed expected \cite{Dine:1985he, Ooguri:2006in, Ooguri:2018wrx, Hebecker:2018vxz}.
Axions, which fulfill perturbative shift symmetries, appear only derivatively in the action and hence they do not contribute to the potential.
However, the couplings of the axion kinetic terms to the scalars are not forbidden by the shift symmetries.
The form of axion kinetic terms in eq.~(\ref{axion-scalar kinetic energy}) often appears in string theory and supergravity settings.
Indeed, string compactifications naturally involve a plethora of pseudoscalars whose action is invariant under a perturbative shift symmetry.
Such fields then appear in the effective action only through exponential kinetic couplings to the canonical scalars $\phi^a$.
This is a direct consequence of their origin as Kaluza-Klein zero-modes of higher-dimensional antisymmetric tensor fields, such as the NSNS- and RR-forms, and the shift symmetry is a manifestation of the gauge invariance of the original fields.
Effects such as couplings to localized sources and background fluxes might stabilize these pseudoscalars,\footnote{This allows us to consider the $\Lambda_i$-terms to be constants. In general, they can depend polynomially on such pseudoscalars, but once in the minimum, such fields are constant.
Multi-branched flux potentials for axions appeared in the context of M-theory compactifications on $G_2$-manifolds \cite{Beasley:2002db}. They were later used to generate potentials for axion monodromy inflation \cite{Marchesano:2014mla, Blumenhagen:2014gta, Hebecker:2014eua, McAllister:2014mpa}.} but not all of them. These latter fields are the axions $\zeta^r$ we study in this paper.
In general, however, the kinetic terms may not be diagonal in the original basis of the moduli.
Nonetheless, eqs.~(\ref{axion-scalar kinetic energy}, \ref{generic exponential potential}) cover a sufficiently general class of theories that warrant our detailed study {and a necessary starting point}.

The complete cosmological equations read
\begin{subequations}
\begin{align}
    & \ddot{\zeta}^r - \kappa_d \dot{\zeta}^r \sum_{a=1}^m \lambda_{ra} \dot{\phi}^a + (d-1) H \dot{\zeta}^r = 0, \label{axion-scalar axion FRW-KG eq.} \\
    & \begin{aligned}[b]
    \ddot{\phi}^a + \dfrac{1}{2} \kappa_d \sum_{r=1}^p \lambda_{ra} \, \e^{-\kappa_d \sum_{b} \lambda_{r b} \phi^b} \, (\dot{\zeta}^r)^2 & \\[-1.25ex]
    + (d-1) H \dot{\phi}^a + \dfrac{\der V}{\der \phi^a} & = 0, \end{aligned} \label{axion-scalar scalar FRW-KG eq.} \\[-0.5ex]
    & H^2 = \dfrac{2 \kappa_d^2}{(d-1) (d-2)} \bigl[ T[\phi, \zeta] + V[\phi] \bigr]. \label{axion-scalar Friedmann eq.}
\end{align}
\end{subequations}
A combination of eq.~(\ref{axion-scalar axion FRW-KG eq.}, \ref{axion-scalar scalar FRW-KG eq.}, \ref{axion-scalar Friedmann eq.}) gives the further useful relationship
\begin{equation} \label{axion-scalar acceleration eq.}
    \dot{H} = - \dfrac{2 \kappa_d^2}{d-2} \, T[\phi,\zeta].
\end{equation}
Non-negative potentials constrain the $\epsilon$-parameter as $0 \leq \epsilon \leq d-1$ at all times.
It is known that eqs.~(\ref{axion-scalar axion FRW-KG eq.}, \ref{axion-scalar scalar FRW-KG eq.}, \ref{axion-scalar Friedmann eq.}) can be recast in terms of a constrained autonomous system of first-order ordinary differential equations \cite{Copeland:1997et, Cicoli:2020cfj}: we will exploit this formulation to get sharp bounds on the $\epsilon$-parameter at late times.

A few definitions are in order.
Let $\gamma_i$ and $\lambda_r$ denote the vectors with components $\smash{(\gamma_i)_a = \gamma_{ia}}$ and $\smash{(\lambda_r)_a = \lambda_{ra}}$, for $i=1,\dots,m$ and $r=1,\dots,p$.
Let $\gamma_\infty$ be the minimal-length vector joining the origin to the potential-coupling convex hull: in this article, we always consider potentials with $\smash{(\gamma_\infty)^2>0}$, i.e. potentials with no minimum.
Let $\smash{\gamma_a = \min_i \gamma_{ia}}$ and $\smash{\Gamma_a = \max_i \gamma_{ia}}$, and let $\Gamma(d) = 2 \sqrt{(d-1)/(d-2)}$. All the bounds we will present are basis-independent; indeed, the theory is invariant under $\smash{\mathrm{O}(n)}$-rotations on canonical scalars.

\subsection{Axion-scalar antialignment}
Here we consider coupling vectors that, for all $r$-indices, fulfill the geometric condition
\begin{equation} \label{axion-scalar antialignment condition}
    \gamma_\infty \cdot \lambda_r < 0.
\end{equation}
If $(\gamma_{\infty})^2 \leq \Gamma^2(d)$, then we can demonstrate that, at late-enough times, the $\epsilon$-parameter takes the value
\begin{equation} \label{axion-scalar antialignment: late-time epsilon}
    \epsilon = \dfrac{d-2}{4} \, (\gamma_\infty)^2.
\end{equation}
Moreover, all the axions completely decouple from the dynamics, i.e. $\dot{\zeta}^r=0$, and the field equations have a unique late-time attractor.
This coincides with the attractor of the theory with just the canonical scalars and the potential being truncated to the ones with couplings lying on the hyperplane of the original convex hull that is orthogonal to the vector $\gamma_\infty$. An example is in fig. \ref{fig.: axion-scalar antialignment}.
Rigorous mathematical proofs are in app. \ref{app: late-time scalar-axion cosmologies}: see theorem \ref{theorem: axion-scalar antialignment - X convergence} and corollaries \ref{corollary: axion-scalar antialignment - W convergence}-\ref{corollary: axion-scalar antialignment - xi convergence}.
To motivate this result physically, we can see that the axion kinetic energy and the scalar potential have inverse exponential dependencies on the scalars, and at least one tends to vanish as the scalars evolve to asymptotic values.
As kinetic-energy domination means maximal deceleration, the energy density of a purely-kinating field would fall off over time more quickly than that of a scaling solution, where the kinetic and the potential energy of a scalar fluid are of the same order, hence the bound.
If $\smash{(\gamma_\infty)^2 \geq \Gamma^2(d)}$, then at late times one has the exact value $\smash{\epsilon = d-1}$, which is simply pure kination in the presence of exceedingly steep potentials.
A rigorous proof is in app. \ref{app: late-time scalar-axion cosmologies}: see corollary \ref{corollary: axion-scalar antialignment for large c - xi convergence}.

\begin{figure}[ht]
    \centering
    \begin{tikzpicture}[xscale=0.65,yscale=0.65,every node/.style={font=\normalsize}]
 
    \begin{scope}
        \draw[orange, thick, fill=orange!35!white] (5,2) -- (3.5,4.5) -- (-1.2,4) -- (5,2);
    
        \draw[->, thick, teal] (0,0) -- (5,2) node[right,black]{$\gamma_1$};
        \draw[->, thick, teal] (0,0) -- (3.5,4.5) node[above left,black]{$\gamma_2$};
        \draw[->, thick, teal] (0,0) -- (-1.2,4) node[above,black]{$\gamma_3$};

        \draw[->, thick, cyan] (0,0) -- (-3,0.5) node[left,black]{$\lambda_1$};
        \draw[->, thick, cyan] (0,0) -- (-3.5,-1.5) node[below,black]{$\lambda_2$};
        \draw[->, thick, cyan] (0,0) -- (0.5,-2) node[right,black]{$\lambda_3$};
        \draw[->, thick, cyan] (0,0) -- (1,-1) node[right,black]{$\lambda_4$};

        \draw[rotate=atan(31/10), densely dotted, magenta] (0,-3.3) -- (0,3.3) node[above right]{};

        \draw[->, thick, purple] (0,0) -- (1120/1061,3472/1061) node[right,pos=0.8]{$\gamma_\infty$};
    \end{scope}

    \begin{scope}[yshift=-230pt]
        \draw[orange, thick, fill=orange!35!white, opacity=0.15] (5,2) -- (3.5,4.5) -- (-1.2,4) -- (5,2);
        \draw[orange, thick] (5,2) -- (-1.2,4);
    
        \draw[->, thick, teal] (0,0) -- (5,2) node[right,black]{$\gamma_1$};
        \draw[->, thick, teal, opacity=0.15] (0,0) -- (3.5,4.5) node[above left,black]{$\gamma_2$};
        \draw[->, thick, teal] (0,0) -- (-1.2,4) node[above,black]{$\gamma_3$};

        \draw[->, thick, cyan, opacity=0.15] (0,0) -- (-3,0.5) node[left,black]{$\lambda_1$};
        \draw[->, thick, cyan, opacity=0.15] (0,0) -- (-3.5,-1.5) node[below,black]{$\lambda_2$};
        \draw[->, thick, cyan, opacity=0.15] (0,0) -- (0.5,-2) node[right,black]{$\lambda_3$};
        \draw[->, thick, cyan, opacity=0.15] (0,0) -- (1,-1) node[right,black]{$\lambda_4$};

        \draw[rotate=atan(31/10), densely dotted, magenta, opacity=0.15] (0,-3.3) -- (0,3.3) node[above right]{};

        \draw[->, thick, purple] (0,0) -- (1120/1061,3472/1061) node[right,pos=0.8]{$\gamma_\infty$};
    \end{scope}

    \end{tikzpicture}
    \caption{The top figure shows a convex hull (orange filled shape), with distance from the origin $\gamma_\infty$ (purple vector), generated by potential couplings $\gamma_i$ (teal vectors) and axion couplings $\lambda_r$ (cyan vectors) such that $\smash{\gamma_\infty \cdot \lambda_r < 0}$. The bottom figure shows the relevant part of the canonical-scalar theory that is relevant to computing the late-time $\epsilon$-parameter.}
    \label{fig.: axion-scalar antialignment}
\end{figure}
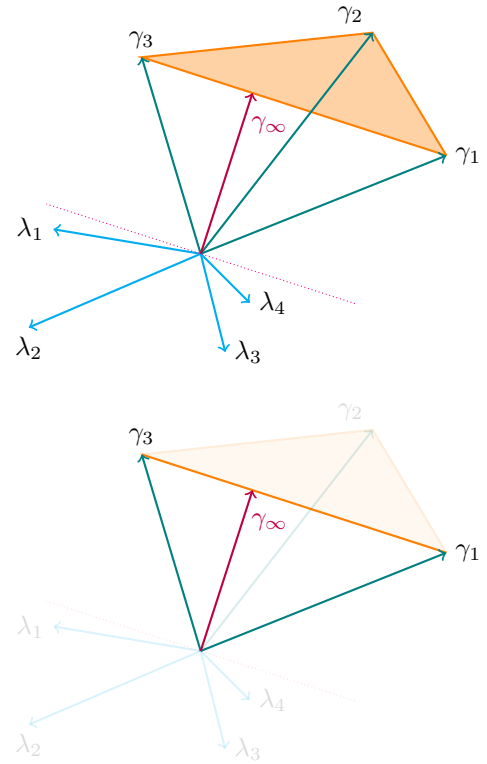

Finally, we emphasize an important implication of eq.~(\ref{axion-scalar antialignment: late-time epsilon}).
In ref.~\cite{Shiu:2023fhb}, in the absence of axions, scaling solutions are proven to be unique late-time attractors if $\smash{\mathrm{rank} \, \gamma_{ia} = m}$, with $m=n$, and if all the potential couplings are such that the vector $\gamma_\infty$ is orthogonal to the coupling convex hull; in particular, see ref.~\cite[ssec. III.B-sssec. IV.C.1]{Shiu:2023fhb}.
It turns out that these extra assumptions are not necessary to analytically prove convergence to the scaling solution.
Indeed, such a setup is a subcase of the result above in the further presence of axions.
For example, see figs. \ref{fig.: axion-scalar antialignment}\footnote{In fig. \ref{fig.: axion-scalar antialignment}, the result stands for the following reason. In the absence of axions, the scalar product in eq.~(\ref{axion-scalar antialignment condition}) is formally $\smash{\gamma_\infty \cdot \lambda_r = 0}$. In this case, corollary \ref{corollary: axion-scalar antialignment - W convergence}, which shows the vanishing of the axion kinetic energy under the assumption $\smash{\gamma_\infty \cdot \lambda_r < 0}$, is not needed since there is no axion kinetic energy in the first place.} and \ref{fig.: general convex-hull criterion}.
More comments are in app. \ref{app: late-time scalar-axion cosmologies}: see remark \ref{remark: paper-2 generalization}.

\begin{figure}[ht]
    \centering
    \begin{tikzpicture}[xscale=0.80,yscale=0.80,every node/.style={font=\normalsize},rotate=35]
    
    \draw[densely dotted, magenta!65!orange] (4,-5) -- (4,0.5);
    \draw[orange, thick, fill=orange!35!white] (4,-4.5) -- (4,-1.5) -- (6,-4) -- (4,-4.5);
    
    \draw[->, thick, teal] (0,0) -- (4,-4.5) node[right,black]{$\gamma_1$};
    \draw[->, thick, purple] (0,0) -- (4,-1.5) node[above,black]{$\gamma_2$} node[above, pos=0.65] {$\gamma_\infty$};
    \draw[->, thick, teal] (0,0) -- (6,-4) node[right,black]{$\gamma_3$};

    \draw[->, magenta!80!orange] (0,0) -- (4,0);
    
    \end{tikzpicture}
    \caption{The figure shows coupling vectors $\gamma_i$ (teal vectors) generating a convex hull (orange filled shape) with distance from the origin $\gamma_\infty$ (purple vector). Asymptotically, only the subset of potentials generating the minimal-distance convex hull is relevant to determining the late-time attractor.}
    \label{fig.: general convex-hull criterion}
\end{figure}
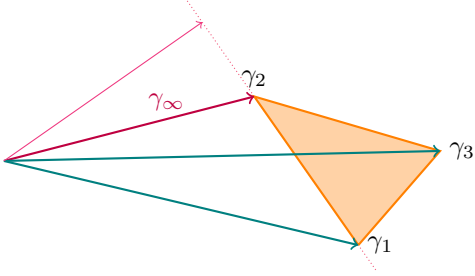

\subsection{Axion-scalar alignment} \label{ssec.: axion-scalar alignment}
Now we consider a different geometric arrangement of the coupling vectors.
To facilitate understanding of our results, we begin by presenting a minimal setup.
Then, we state our results in full generality.

\subsubsection{Minimal setup: 2-dimensional coupling space}
The minimal non-trivial setup we can consider is in a 2-dimensional coupling space, i.e. for $n=2$, with two potential terms and two axions, i.e. for $m=p=2$.
In a given coordinate system, we consider a pair of axion coupling vectors $\lambda_{1,2}$-vectors lying in the first quadrant.
Moreover, we consider a pair of linearly-independent potential coupling $\gamma_{1,2}$-vectors that are also in the first quadrant and that fall within the convex cone $\mathrm{C}(\lambda_1,\lambda_2)$ generated by the axion couplings.

The easiest of such configurations is for orthogonal pairs of vectors, as in fig. \ref{fig.: SUSY assisted inflation}.
In this case, if $\smash{(\Gamma)^2 \leq \Gamma^2(d)}$, where $\Gamma$ is the vector with components $\smash{\Gamma_a}$, then we can show analytically that the late-time $\epsilon$-parameter is bounded as $\epsilon \leq [(d-2)/4] \, (\Gamma)^2$.

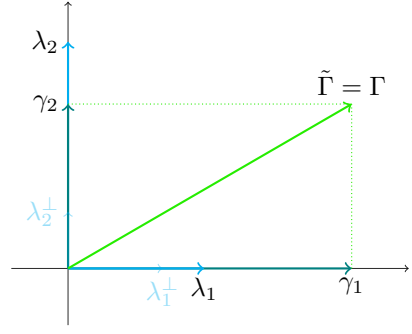
\begin{figure}[ht]
    \centering
    \begin{tikzpicture}[xscale=1.50,yscale=1.50,every node/.style={font=\normalsize},rotate=0]
    
    
    \draw[->,ultra thin,opacity=0.85] (-0.5,0) -- (3.0,0);
    \draw[->,ultra thin,opacity=0.85] (0,-0.5) -- (0,2.35);
    
    \draw[->, thick, teal] (0,0) -- (2.5,0) node[below,black]{$\gamma_1$};

    \draw[->, thick, cyan] (0,0) -- (6/5,0) node[below,black]{$\lambda_1$};
    \draw[->, thick, cyan] (0,0) -- (0,2) node[left,black]{$\lambda_2$};
    
    \draw[->, thick, teal] (0,0) -- (0,1.45) node[left,black]{$\gamma_2$};

    \draw[->, thin, cyan, opacity=0.4] (0,0) -- (5/6,0) node[below]{$\lambda^\perp_1$};
    \draw[->, thin, cyan, opacity=0.4] (0,0) -- (0,1/2) node[left]{$\lambda^\perp_2$};
    
    \draw[densely dotted, green!85!orange] (2.5,0) -- (2.5,1.45);
    \draw[densely dotted, green!85!orange] (0,1.45) -- (2.5,1.45);
    \draw[->, thick, green!85!orange] (0,0) -- (2.5,1.45) node[above,black]{$\tilde{\Gamma} = \Gamma$};
    
    
    \end{tikzpicture}
    \caption{The figure shows a parameter-space configuration with orthogonal pairs of axion (cyan) and potential (teal) coupling vectors. In this case, the $\Gamma$-vector (green) has a squared length equal to the square of the diagonal generated by the potential couplings.}
    \label{fig.: SUSY assisted inflation}
\end{figure}

If the vectors are not orthogonal, we define two vectors $\smash{\lambda^\perp_{1,2}}$ that have inverse length and positive scalar product with respect to $\lambda_{1,2}$, and that are orthogonal to the other vector $\lambda_{2,1}$, i.e. $\smash{\lambda_{1,2} \cdot \lambda^\perp_{1,2} = 1}$ and $\smash{\lambda_{1,2} \cdot \lambda^\perp_{2,1} = 0}$.
Note that such vectors have components $\smash{(\lambda^\perp_r)_a = [(\lambda^{-1})^{\mathrm{t}}]^{ra}}$.
An example is in fig. \ref{fig.: axion-scalar alignment}.
Let $\Pi_r$ be the maxima of the projections of the vectors $\gamma_{1,2}$ onto the vectors $\smash{\lambda^\perp_r}$, i.e. $\smash{\Pi_r = \max_i \, (\gamma_i \cdot \lambda^\perp_r)}$.
In our geometric setup, we can see that $\Pi_r \geq 0$. Now, we define $\tilde{\Gamma}$ to be the vector with components $\smash{\tilde{\Gamma}_a = \sum_{r=1}^2 \Pi^r \lambda_{ra}}$.
To visualize the meaning of this vector, one can again think of the case where the vectors $\lambda_r$ are orthogonal, as in fig. \ref{fig.: SUSY assisted inflation}: in this case, it is clear that $\smash{\tilde{\Gamma} = \Gamma}$.
Of course, if the $\lambda_r$-vectors are not orthogonal, this picture can be modified depending on how large the deviation from right angles is.
We can now state our bound, which is extremely simple in terms of the definitions above.
If $\smash{\Gamma \cdot \tilde{\Gamma} \leq \Gamma^2(d)}$, then the late-time $\epsilon$-parameter is bounded as $\smash{\epsilon \leq [(d-2)/4] \, (\Gamma \cdot \tilde{\Gamma})}$.
If both vectors $\gamma_1$ and $\gamma_2$ are very close to the bisector of the angle between $\lambda_1$ and $\lambda_2$, then $\smash{\Gamma \cdot \tilde{\Gamma}}$ is close to $(\Gamma)^2$.

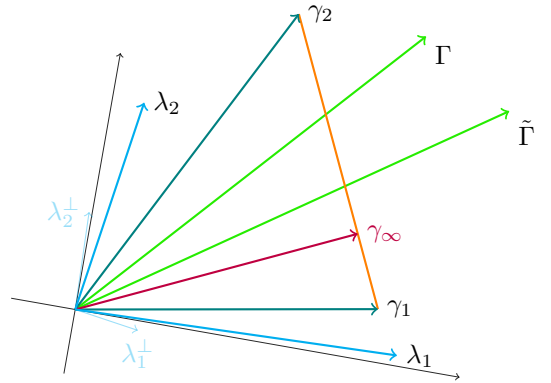
\begin{figure}[ht]
    \centering
    \begin{tikzpicture}[xscale=0.86,yscale=0.86,every node/.style={font=\normalsize},rotate=15]

    \begin{scope}[rotate=-25]
        \draw[->,ultra thin,opacity=0.85] (-1,0) -- (6,0);
        \draw[->,ultra thin,opacity=0.85] (0,-1) -- (0,4);
            
            \draw[->, thick, green!85!orange] (0,0) -- (4.586,5.074) node[below right,black]{$\Gamma$};
            \draw[->, thick, green!85!orange] (0,0) -- (6.045,4.164) node[below right,black]{$\tilde{\Gamma}$};
    \end{scope}

    \draw[orange, thick] (4.5,-1.2) -- (4.5,3.5);
    
    \draw[->, thick, teal] (0,0) -- (4.5,-1.2) node[right,black]{$\gamma_1$};
    \draw[->, thick, teal] (0,0) -- (4.5,3.5) node[right,black]{$\gamma_2$};
    \draw[->, thick, purple] (0,0) -- (4.5,0) node[right] {$\gamma_\infty$};

    \begin{scope}[rotate=30]
    
        \draw[->, thick, cyan] (0,0) -- (3,-4) node[right,black]{$\lambda_1$};
        \draw[->, thick, cyan] (0,0) -- (3,3/2) node[right,black]{$\lambda_2$};

        \draw[->, thin, cyan, opacity=0.4] (0,0) -- (5/11,-10/11) node[below]{$\lambda^\perp_1$};
        \draw[->, thin, cyan, opacity=0.4] (0,0) -- (40/33,10/11) node[left]{$\lambda^\perp_2$};
        
    \end{scope}
    
    \end{tikzpicture}
    \caption{The figure shows a parameter-space configuration such that, for all the axion couplings $\lambda_r$ (cyan), the projections of the orthogonal couplings $\smash{\lambda^\perp_r}$ (fading cyan) on all potential couplings $\gamma_i$ (teal) are non-negative. The minimal-distance vector $\gamma_\infty$ (purple) from the origin to the potential convex hull (orange) is also shown. An orthogonal coordinate system is also shown to facilitate identifying right angles.}
    \label{fig.: axion-scalar alignment}
\end{figure}

Depending on the geometry of the couplings, we can also identify lower bounds for the $\epsilon$-parameter. 
After presenting minimal scenarios with the alignment of axion-scalar couplings up to here, we discuss the most general geometric alignment scenario that we can bound below.

\subsubsection{General axion-scalar alignment}
We consider couplings with $\mathrm{rank} \, \gamma_{ia} = \mathrm{rank} \, \lambda_{ra} = n$, with $m, p \geq n$, which means the matrices $\gamma_{ia}$ and $\lambda_{ra}$ are left-invertible, and we define $p$ vectors $\smash{\lambda^\perp_r}$ with components $\smash{(\lambda^\perp_r)_a = [(\lambda^{-1})^{\mathrm{t}}]^{ra}}$. Let $\smash{\pi^r = \min_{i} \, \sum_{a=1}^n \gamma_{ia} \lambda^\perp_{ra}}$ and $\smash{\Pi^r = \max_{i} \, \sum_{a=1}^n \gamma_{ia} \lambda^\perp_{ra}}$.
Furthermore, we consider parameter spaces such that
\begin{equation} \label{axion-scalar alignment condition 1}
    \pi_r \geq 0
\end{equation}
for all $r$-indices and such that all vectors $\smash{\gamma_i}$ and $\smash{\lambda_r}$ have non-negative components, i.e.
\begin{equation} \label{axion-scalar alignment condition 2}
    \gamma_{ia}, \, \lambda_{ra} \geq 0
\end{equation}
for all $i$-, $r$- and $a$-indices.
Such conditions are easy to check in explicit examples: see e.g. figs. \ref{fig.: SUSY assisted inflation} and \ref{fig.: axion-scalar alignment}.
After defining a vector $\tilde{\Gamma}$ with components $\smash{\tilde{\Gamma}_a = \sum_{r=1}^p \Pi^r \lambda_{ra}}$, we are able to prove two analytic bounds.
\begin{itemize}[leftmargin=*]
    \item If all vectors $\gamma_i$ and $\lambda_r$ are such that
    \begin{equation} \label{axion-scalar alignment extra condition A}
        (\gamma_i + \lambda_r) \cdot \tilde{\Gamma} < \Gamma^2(d),
    \end{equation}
    then, at late-enough times, the $\epsilon$-parameter is such that
    \begin{equation} \label{axion-scalar alignment, A: late-time epsilon}
        \dfrac{d-2}{4} \, (\gamma_\infty)^2 \leq \epsilon \leq \dfrac{d-2}{4} \, (\tilde{\Gamma})^2.
    \end{equation}
    In this case, all axion kinetic energies vanish, i.e. at late-enough times one has $\smash{\dot{\zeta}^r = 0}$.
    \item If the inequality holds
    \begin{equation} \label{axion-scalar alignment extra condition B}
        \Gamma \cdot \tilde{\Gamma} \leq \Gamma^2(d),
    \end{equation}
    then the late-time $\epsilon$-parameter is bounded as
    \begin{equation} \label{axion-scalar alignment, B: late-time epsilon}
        \epsilon \leq \dfrac{d-2}{4} \, (\Gamma \cdot \tilde{\Gamma}).
    \end{equation}
    This is the general form of the bounds introduced in the previous subsubsection.
\end{itemize}
Both these bounds can be optimized in the bases that minimize the values of $\smash{(\tilde{\Gamma})^2}$ and $\smash{(\Gamma \cdot \tilde{\Gamma})}$.
Analytic proofs of these results are in app. \ref{app: late-time scalar-axion cosmologies}: see theorems \ref{theorem: axion-scalar alignment - xi bounds} and \ref{theorem: axion-scalar alignment - xi general bound}.

\subsection{Partial coupling misalignment}

We can prove one more result even if all the potential and kinetic couplings are in the same half of the coupling hyperplane.

Let $\smash{\Lambda_a = \max_r \, \lambda_{ra}}$.
If $\Lambda_a \leq 0$ and $\gamma_a + \Lambda_a \geq \Gamma(d)$ for at least one field $\phi^a$, then we can show that at late-enough times one has
\begin{equation} \label{large-gamma: late-time epsilon}
    \epsilon = d-1,
\end{equation}
which corresponds to a pure (semi-)eternal kination phase.
This is proven analytically in app. \ref{app: late-time scalar-axion cosmologies}: see corollary \ref{corollary: xi-limit for L<0 and c+L>1}.
Although the requirements for eq.~(\ref{large-gamma: late-time epsilon}) might seem restrictive, we can seek any scalar-field basis rotation to fulfill the condition as usual.
So, loosely speaking, even just one long $\gamma_i$-vector (compared to $\Gamma(d)$) with a small overlap with the $\lambda_r$-vectors induces a late-time solution with the maximal $\epsilon$-parameter $\epsilon=d-1$.
An example is in fig. \ref{fig.: large-gamma: late-time epsilon}; fig. \ref{fig.: axion-scalar antialignment} can also apply.
We can focus on the specific scalar $\smash{\phi_{a}}$ to motivate this result.
Independently of all the other fields, the axion kinetic term and the scalar potential scale as $\smash{T[\zeta] \geq \tilde{T}[\zeta] \, \e^{-\kappa_d \Lambda_{\underaccent{\dot}{a}} \phi_{\underaccent{\dot}{a}}}}$ and $\smash{V[\phi] \leq \tilde{V}[\tilde{\phi}] \, \e^{-\kappa_d \gamma_{\underaccent{\dot}{a}} \phi_{\underaccent{\dot}{a}}}}$, with no index summation meant by the underdotted index and where the tilde-accent means that all other $\smash{\phi^{\underaccent{\dot}{a}}}$-dependencies have been removed.
Hence, the field $\phi_a$ necessarily leads to late-time kination because the kinetic energy grows indefinitely with respect to the potential (as $\smash{\Lambda_{\underaccent{\dot}{a}} \phi_{\underaccent{\dot}{a}} \leq 0}$), or because the potential itself is steep enough to become asymptotically irrelevant (as $\smash{\gamma_a \geq \Gamma(d)}$; see e.g. ref.~\cite{Shiu:2023nph}).

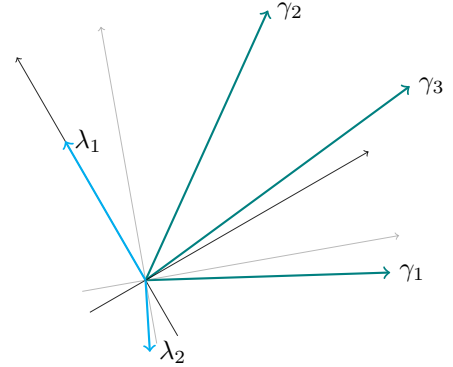
\begin{figure}[ht]
    \centering
    \begin{tikzpicture}[xscale=0.85,yscale=0.85,every node/.style={font=\normalsize},rotate=25]

    \draw[->, thick, teal] (0,0) -- (3.5,-1.5) node[right,black]{$\gamma_1$};
    \draw[->, thick, teal] (0,0) -- (3.5,3) node[right,black]{$\gamma_2$};
    \draw[->, thick, teal] (0,0) -- (5,1) node[right,black]{$\gamma_3$};

    \begin{scope}[rotate=-15]    
        \draw[->,ultra thin,opacity=0.55,gray] (-1,0) -- (4,0);
        \draw[->,ultra thin,opacity=0.55,gray] (0,-1) -- (0,4);
    \end{scope}

    \begin{scope}[rotate=5]    
        \draw[->,ultra thin,opacity=0.85] (-1,0) -- (4,0);
        \draw[->,ultra thin,opacity=0.85] (0,-1) -- (0,4);
        
        \draw[->, thick, cyan] (0,0) -- (0,2.5) node[right,black]{$\lambda_1$};
        \draw[->, thick, cyan] (0,0) -- (-0.5,-1) node[right,black]{$\lambda_2$};
    \end{scope}
    
    \end{tikzpicture}
    \caption{The figure shows a parameter-space configuration in which all the potential couplings $\gamma_i$ (teal) are much larger than the (non-positive) axion couplings $\lambda_r$ (cyan), in an appropriate field-space basis. Two orthogonal coordinate systems (fading gray and black) are also added.}
    \label{fig.: large-gamma: late-time epsilon}
\end{figure}

\subsection{Comments on axions with a potential}
As mentioned before, some of the fields $\zeta^r$ might also appear in the potential. In this case, eq. (\ref{axion-scalar axion FRW-KG eq.}) must be replaced by
\begin{align*}
    \ddot{\zeta}^r - \kappa_d \dot{\zeta}^r \sum_{a=1}^m \lambda_{ra} \dot{\phi}^a + (d-1) H \dot{\zeta}^r + \e^{\kappa_d \sum_{b} \lambda_{r b} \phi^b} \dfrac{\der V}{\der \zeta^r} = 0,
\end{align*}
and the fields $\zeta^r$ also appear in eq. (\ref{axion-scalar Friedmann eq.}) through the potential $V=V[\phi, \zeta]$.
For instance, the $\Lambda_i$-terms may be polynomial in the fields $\zeta^r$.
If the fields $\zeta^r$ have sufficiently small initial velocities, then our bounds in eqs. (\ref{axion-scalar antialignment: late-time epsilon}) and (\ref{large-gamma: late-time epsilon}) can be argued to be still in place.
Indeed, as the scalars $\phi^a$ start to evolve towards the asymptotics, the term involving the potential gradient in the equation above gets an exponential suppression of the form $\smash{\sum_i \Lambda_i(\zeta) \, \e^{\kappa_d \sum_{b} (\lambda_{r b} - \gamma_{i b}) \phi^b} \leq \sum_i \Lambda_i(\zeta) \, \e^{- \kappa_d \sum_{b}\gamma_{i b} \phi^b}}$, when the $\lambda_r$- and $\gamma_i$-vectors are all oriented in opposing directions.

\section{Examples} \label{sec.: examples}
In this section, we illustrate a few simple examples of where to see our bounds at work.

\subsection{Phase-space studies in low field-space dimension}
To test our bounds, let us start with examples with low field-space dimensions.
For $n=p=1$, and taking $m=1$, the critical points for $d=4$ can be read off in refs.~\cite{Cicoli:2020cfj, Cicoli:2020noz}.
Fixing $\gamma>0$ without loss of generality, the conditions that ensure perturbative stability of critical points with $\epsilon \leq \gamma^2/2$, and no axion kinetic energy, are:
(i) $\lambda<0$, which indeed reduces to the solution in eq.~(\ref{axion-scalar antialignment: late-time epsilon});
(ii) if $\lambda>0$, $\smash{\gamma < \sqrt{\lambda^2/4+6} - \lambda/2}$ or $\gamma > \lambda$, which both encompass the requirements for eqs. (\ref{axion-scalar alignment, A: late-time epsilon}, \ref{axion-scalar alignment, B: late-time epsilon}, \ref{large-gamma: late-time epsilon}).\footnote{Both refs.~\cite{Cicoli:2020cfj, Cicoli:2020noz} also consider the presence of an additional fluid with the constant equation of state $p/\rho = q-1$. This does not alter qualitatively our conclusions since the only change that we expect is that, if $q$ is small enough, this will be the late-time attractor, with $\epsilon$-parameter $\epsilon = 3q/2 \leq \gamma^2/2$.}
See also ref.~\cite{Revello:2023hro} for a study of linear stability of theories with diagonal kinetic-coupling matrix and one potential term.
Further phase-space analyses for $n=1,2$, $p=1$, and $m=1$ are in refs.~\cite{Sonner:2006yn, Russo:2022pgo}.

\subsection{STU-models of (generalized) assisted inflation}
A simple example to discuss is that of assisted inflation \cite{Liddle:1998jc} in which both the axion and the potential coupling matrices are diagonal, i.e. $\smash{\lambda_{ra} = \delta_{ra} \, \lambda_r}$ and $\smash{\gamma_{ia} = \delta_{ia} \, \gamma_i}$, with $n=m=p$.
For instance, this can be the case for the so-called STU-models.
For simplicity, we take $n=2$ and refer to fig. \ref{fig.: SUSY assisted inflation}.
Employing a vector notation for clarity, we write the couplings of the theory as $\smash{\underline{\lambda}{}_{1,2} = \lambda_{1,2} \, \hat{\theta}_{1,2}}$ and $\smash{\underline{\gamma}{}_{1,2} = \gamma_{1,2} \, \hat{\theta}_{1,2}}$, where $\smash{\hat{\theta}_a}$ are the unit vectors in coupling space, for $a=1,2$.
One then finds the orthogonal couplings $\smash{\lambda^\perp_{1,2} = \hat{\theta}_{1,2}/\lambda_{1,2}}$, with the minimum and maximum projections $\smash{\pi^{1,2} = 0}$ and $\smash{\Pi^{1,2} = \gamma_{1,2}/\lambda_{1,2}}$.
This means that $\smash{\tilde{\underline{\Gamma}} = \underline{\Gamma}}$.
In particular, note that $\smash{(\tilde{\underline{\Gamma}})^2 = (\underline{\Gamma} \cdot \tilde{\underline{\Gamma}}) = (\underline{\Gamma})^2}$.
If for example $\smash{(\underline{\Gamma})^2 \leq \Gamma^2(d)}$, then, by eq.~(\ref{axion-scalar alignment, B: late-time epsilon}), the late-time $\epsilon$-parameter is bounded as $\smash{\epsilon \leq [(d-2)/4] \, (\underline{\Gamma})^2}$.
Generalizations to higher numbers of couplings are immediate.

As an instance, string compactifications exist such that the low-energy effective theory, in dimension $d=4$, is encoded in a Kähler potential
\begin{align*}
    \kappa_4^2 K = - \sum_{a} n_{a} \ln \, [-\I \, (\xi^{a} - \overline{\xi}{}^{a})] - \sum_{l} q_l \ln \, [-\I \, (\chi^l - \overline{\chi}{}^l)],
\end{align*}
where $\xi^{a}$ and $\chi^l$ are chiral multiplets, for some constants $n_a$ and $l_l$.
Suppose that, after the stabilization of the fields $\chi^l$, the scalar potential happens to take the form
\begin{align*}
    V = \sum_{a} \dfrac{V_a}{[-\I \, (\xi^{a} - \overline{\xi}{}^{a})]^{m_a}},
\end{align*}
for some positive constants $V_a$ and $m_a$. In the parametrization
\begin{align*}
    \xi^a = \dfrac{\sqrt{2}}{\sqrt{n^a}} \, \kappa_4 \zeta^a + \I \, \e^{\frac{\sqrt{2}}{\sqrt{n^a}} \, \kappa_4 \phi^a},
\end{align*}
the theory is precisely in the formulation of eqs. (\ref{axion-scalar kinetic energy}, \ref{generic exponential potential}), with the diagonal couplings $\smash{\lambda_{ra} = (2 \sqrt{2}/\sqrt{n^a}) \, \delta_{ra}}$ and $\smash{\gamma_{ia} = (\sqrt{2} \, m^a/\sqrt{n^a}) \, \delta_{ia}}$.
Single-multiplet examples exist for e.g. $m = n = 1,2,3$ \cite{Font:1990nt, Cicoli:2011it, Cicoli:2016xae, Cicoli:2017axo, Saltman:2004sn}, as has been considered in ref. \cite{Brinkmann:2022oxy}.
Then, we might conclude with the late-time bound
\begin{align*}
    \epsilon \leq \sum_a \dfrac{(m^a)^2}{n^a}
\end{align*}
beyond the linear stability analysis and for arbitrary numbers of fields.
One can also generalize all the above statements to the case of generalized assisted inflation \cite{Copeland:1999cs}.
In STU-models, the $\gamma_{ia}$-couplings can for instance come from the exponential of the Kähler potential.
Although not necessarily orthogonal, such $\gamma_{ia}$-couplings still fit into one coupling-space hyperquadrant.

For instance, one could study limits with large complex-structure moduli arising from F-theory compactifications on Calabi-Yau fourfolds \cite{Grimm:2019ixq, Calderon-Infante:2022nxb, Revello:2023hro} or non-geometric compactifications with no Kähler moduli \cite{Becker:2006ks, Bardzell:2022jfh, Cremonini:2023suw}.
In view of the results we presented, the study of asymptotic regions might be easier in that not all the axions need to be stabilized in order to have analytic control on the late-time universe, and arbitrary numbers of axion kinetic terms and potential terms might be considered.
If for example $\smash{(\underline{\gamma}{}_i + \underline{\lambda}{}_i) \cdot \underline{\Gamma} \leq \Gamma^2(d)}$ for all $i$- and $r$-indices, then, by eq.~(\ref{axion-scalar alignment, A: late-time epsilon}), the late-time $\epsilon$-parameter is bounded as $\smash{[(d-2)/4] \, (\gamma_\infty)^2 \leq \epsilon \leq [(d-2)/4] \, (\underline{\Gamma})^2}$.

Compactifications in which axion-saxion kinetic terms and (part of) the scalar potential take the form in eqs.~(\ref{axion-scalar kinetic energy}, \ref{generic exponential potential}) also exist in KKLT- and LVS-scenarios \cite{Kachru:2003aw, Balasubramanian:2005zx}.
The typical non-perturbative potentials (which become double exponentials in the frame where the scalars are canonically normalized) are convex functions and thus lead to bounds of a similar form to those found in this article.
Moreover, some sectors of the theory may not need to be stabilized by fluxes or non-perturbative or $\alpha'$-corrections, but rather they could be studied for quintessence-like realizations.

We hope to come back to all these applications and more in the future.

\subsection{A string-theoretic toy model}

As another application, we might consider a string-theoretic toy model.
In a type-IIA compactification over a 4-space of vanishing curvature and with $H_3$-flux, the 6-dimensional dilaton $\tdelta$, and the string-frame radion $\tsigma$ are subject to the positive-definite scalar potential
\begin{align*}
    V = \Lambda_{H_3} \, \e^{\kappa_6 \tdelta - 3 \kappa_6 \tsigma}.
\end{align*}
If only the RR-axions corresponding to the $C_1$- and $C_3$-forms, say $\smash{\zeta_2}$ and $\smash{\zeta_4}$, respectively, are part of the effective field theory, then their kinetic action reads
\begin{align*}
    T[\zeta_2, \zeta_4] = \dfrac{1}{2} \, \e^{2 \kappa_6 \tdelta + \kappa_6 \tsigma} \, (\dot{\zeta}_2)^2 + \dfrac{1}{2} \, \e^{2 \kappa_6 \tdelta - \kappa_6 \tsigma} \, (\dot{\zeta}_4)^2.
\end{align*}
One can perform a $\mathrm{SO}(2)$-rotation in the scalar-field basis by angle $\theta = \pi - \arctan \, (1/2)$, defining $\smash{\xi = (2 \tdelta - \tsigma)/\sqrt{5}}$ and $\smash{\eta = (\tdelta + 2 \tsigma)/\sqrt{5}}$.
This maps the potential and the axion kinetic term into the forms
\begin{align*}
    V[\xi, \eta] & = \Lambda_{H_3} \, \e^{\sqrt{5} \, \kappa_6 \xi - \sqrt{5} \, \kappa_6 \eta}, \\
    T[\zeta_2, \zeta_4] & = \dfrac{1}{2} \, \e^{\frac{3}{\sqrt{5}} \kappa_6 \xi + \frac{4}{\sqrt{5}} \kappa_6 \eta} \, (\dot{\zeta}_2)^2 + \dfrac{1}{2} \, \e^{\sqrt{5} \, \kappa_6 \xi} \, (\dot{\zeta}_4)^2.
\end{align*}
In this basis, one has $\smash{\Lambda_{\eta} = 0}$ and $\smash{\gamma_{\eta} + \Lambda_{\eta} = \sqrt{5} = \Gamma(6)}$.
Hence, by eq.~(\ref{large-gamma: late-time epsilon}), we conclude that the late-time $\epsilon$-parameter reads exactly $\smash{\epsilon = 5}$.

\section{Critical points} \label{sec.: critical points}
Although our analytic results do not require knowledge of the critical points of the autonomous system in which the cosmological equations can be recast, we list below such critical solutions and comment on them.

\subsection{Exact solutions}
Let $\smash{\mathrm{rank} \, \gamma_{ia} = m}$ and $\smash{\mathrm{rank} \, \lambda_{ra} = p}$; such conditions can always be imposed by formally truncating the potential and/or the axion terms if needed.
Then, eqs.~(\ref{axion-scalar axion FRW-KG eq.}, \ref{axion-scalar scalar FRW-KG eq.}, \ref{axion-scalar Friedmann eq.}) admit solutions in which the $\epsilon$-parameter is constant and, for each given $i_*$-index, it takes the form
\begin{equation} \label{critical-point epsilon}
    \epsilon_* = \dfrac{d-1}{1 + \rho_*},
\end{equation}
where
\begin{equation} \label{rho}
    \rho_* = \dfrac{1}{\displaystyle \sum_{a=1}^n \sum_{r=1}^p \gamma_{i_*a} (\lambda^{-1})^{ar}}.
\end{equation}
Fixing the initial time as $t_0 = 1/[\epsilon_* H(t_0)]$, we also have
\begin{subequations}
    \begin{align}
        \phi_*^a(t) & = \phi_0^a + \dfrac{2 \rho_*}{\kappa_d} \, \sum_{r=1}^p (\lambda^{-1})^{ar} \; \mathrm{ln} \, \dfrac{t}{t_0}, \label{axion-scalar scaling solution - scalar} \\
        \zeta_*^r (t) & = \zeta_0 + \dfrac{\sqrt{1 + \rho_*}}{\kappa_d \sqrt{\rho_*}} \dfrac{\sqrt{d-2}}{\sqrt{d-1}} \, \upsilon_*^r \, \biggl[\biggl( \dfrac{t}{t_0} \biggr)^{\! \rho_*} \!\!\!\!- 1 \biggr], \label{axion-scalar scaling solution - axion}
    \end{align}
\end{subequations}
where
\begin{align*}
    (\upsilon_*^r)^2 = \sum_{a=1}^n (\lambda^{-1})^{ar} \biggl[ \gamma_{i_* a} - 4 \, \dfrac{d-1}{d-2} \, \dfrac{\rho_*}{1 + \rho_*} \sum_{s=1}^p (\lambda^{-1})^{as} \biggr].
\end{align*}
Of course, all non-zero kinetic and potential energy terms fall over time with a $(1/t^2)$-dependence.\footnote{By plugging the solutions of eqs.~(\ref{axion-scalar scaling solution - scalar}, \ref{axion-scalar scaling solution - axion}) into eq.~(\ref{axion-scalar scalar FRW-KG eq.}) we also find the compatibility condition
\begin{align*}
    \biggl[ \dfrac{1}{2} \, \dfrac{d-2}{d-1} \, \rho_* (1+\rho_*) - \kappa_d^2 t_0^2 \, \Lambda_{i_*} \, \e^{- \kappa_d \sum_{b} \gamma_{i_* b} \phi_0^b} \biggr] \, \gamma_{i_* a} = 0.
\end{align*}}
A derivation of such critical solutions in terms of the autonomous system is in app. \ref{app: late-time scalar-axion cosmologies}.
For diagonal kinetic-coupling matrices and a single potential term, critical solutions are discussed in refs.~\cite{Cicoli:2020cfj, Revello:2023hro, Seo:2024qzf}.
Of course, other critical solutions exist, which one can compute by formally setting to zero the axion terms; these have been studied in ref.~\cite{Collinucci:2004iw}.

Looking at eqs.~(\ref{critical-point epsilon}, \ref{rho}), the critical solutions feature cosmic acceleration, i.e. $\epsilon_* < 1$, if $\rho_* > d-2$.
Geometrically, this means that cosmic acceleration is favoured in the presence of small amounts of overlaps between the potential and the axion couplings (since, in general, the inverse of the $\lambda_{ra}$-matrix provides orthogonal vectors to the $\lambda_r$-vectors, with inverse length), and for long $\gamma_{i}$-vectors and short $\lambda_r$-vectors.\footnote{In the presence of both kinetic and potential couplings, the geometric interpretation of eqs.~(\ref{critical-point epsilon}, \ref{rho}) is not as neat as for only potential couplings \cite{Shiu:2023fhb}.
Thus, it is natural that
any intuition is more complex in nature.
Yet, we have some understanding e.g. for $m=n=p=2$, as discussed in ssec.~\ref{ssec.: axion-scalar alignment}.}
Note that formally $\epsilon_* < 0$ if $\rho_* < -1$; this is an unphysical solution.

\subsection{A speculation}
Based on observations in models where full analytic convergence is known, we propose a speculation on convergence.
In all the examples we know of, there seems to be a general pattern, which we state as follows.

All FLRW-cosmologies involving scalars, axions, and fluids with constant equations of state, if coupled through multi-exponential potentials and multi-exponential axion couplings, feature a plethora of critical-point solutions.
Only a subset of these is physical.\footnote{By this we mean that some solutions of the autonomous system cannot be physical solutions. For instance, for a positive-definite potential, some of the critical points may correspond to a negative-definite potential. Such solutions are unphysical.}
Among the physical solutions, in theories satisfying the weak energy condition, we conjecture that the late-time attractor is always the solution with the smallest of the $\epsilon$-parameter.

Currently, we are unable to formulate this statement more precisely and to prove it in full generality; furthermore, whether what we believe to be an attractor is a stable node or a stable spiral is instead harder to decipher.
However, our claim is confirmed by all the theories in which there is non-perturbative, perturbative and/or numerical control; see e.g. \cite{Copeland:1997et, Collinucci:2004iw, Hartong:2006rt, Shiu:2023fhb, VanRiet:2023cca, Andriot:2023wvg, Revello:2023hro, Andriot:2024jsh}.
There is a very simple physical argument to support our speculation.
A fluid with a constant equation-of-state parameter $q$ has an energy density evolving as $\rho(t) = \rho(t_0) \, (a(t)/a_0)^{(d-1) q}$.
Therefore, the fluid whose energy density dominates asymptotically is the one with the smallest value of $q$.
Because for a single-fluid cosmology, the $\epsilon$-parameter is related to $q$ as $\epsilon = (d-1) \, q/2$, our argument follows through.
All this is of course harder to prove for scalars, with non-constant equations of state.
However, typically their equation of state changes over time in a way that becomes constant asymptotically.
The restriction to satisfying the weak energy condition ensures that one avoids instabilities driving the universe towards contraction or developing phantom fields.
We hope to come back to this general problem in the future.

\section{Discussion}
In this article, we proved analytic bounds for axion-scalar cosmologies in which the potential and the axion coupling vectors (i) lie in opposite halves of the coupling space, or (ii) in the same hyperquadrant, with (iii) further results in case of even partial coupling misalignment.
Our findings are novel, as the previous studies have not examined the coupled axion-scalar systems in such generality.
Furthermore, our results are model-independent, though one can easily apply the criteria we found to specific models. 
Therefore, the universal bounds we derived are of relevance to future studies of dark energy, both from a phenomenological model-building point of view and in embedding such models into string theory.
Such bounds illustrate that there are classes of axion-scalar theories with hyperbolic field spaces that generate the same late-time cosmology as those with zero field-space curvature.
In all these cases, steep potentials prevent asymptotic cosmic acceleration: the multi-field no-go results of ref.~\cite{Shiu:2023nph} still apply.
For diagonal kinetic-coupling matrices and a single potential term, there exist previous perturbative and numerical checks: see e.g. refs.~\cite{Sonner:2006yn, Cicoli:2020cfj, Cicoli:2020noz, Russo:2022pgo, Brinkmann:2022oxy, Revello:2023hro}.
We showed how these previously isolated results follow from the general analysis presented here.
The bounds might as well be seen in a positive light: if one were to find top-down models providing small potential couplings, one would have small $\epsilon$ without concerns with stabilizing the axions.

Our analysis is fully non-perturbative and applicable to arbitrary field-space dimensions.
Furthermore, our results give a clear and definite geometric intuition of the parameter-space configuration (of axion and scalar couplings) that needs to be scrutinized next: if the potential and the axion coupling vectors lie in several hyperquadrants, or if they are aligned but with comparable length, different bounds are in place.
This is where one might hope to find asymptotic cosmic acceleration.
Evidence with small numbers of fields \cite{Cicoli:2020cfj} suggests that these might be scenarios featuring late-time cosmic acceleration also in higher field-space-dimensional setups arising in string compactifications.

Another concrete result we presented is the proof that, in flat moduli spaces, the scaling cosmologies of ref.~\cite{Collinucci:2004iw} are the universal late-time attractors independently of the details of the multi-exponential potential, as long as it does not accommodate a de Sitter minimum. 
This proof relaxes an assumption in our previous work \cite{Shiu:2023fhb}.

Our conclusions are also a testimony to the generality of our methods: we do not need to find explicit solutions to the cosmological equations (i.e. the critical points of the autonomous system) to prove strong analytic bounds on the rate of cosmic acceleration.
The ambitious research program of characterizing fully generally the late-time cosmology of multi-field systems we set out to pursue has led to a few successful results.
For example, the bounds we obtained in this paper provide clear geometric intuition that may point us to possible paths forward.
Still, given the wilderness of the space of cosmological multi-field systems, not all possible parameter-space configurations are covered.
It is our goal to pursue a full analytic description of the late-time cosmology of theories with axions, scalars, and barotropic fluids in future work \cite{axionscalarfluidcosmologies}.
We also believe the lessons learned from our methods might help further progress in different cosmological models, such as those presented recently in refs. \cite{Dienes:2021woi, Marconnet:2022fmx, Andriot:2023wvg, Dienes:2023ziv, Gomes:2023dat, Gallego:2024gay, Andriot:2024jsh, Dienes:2024wnu, Casas:2024xqy, Alestas:2024gxe}.

Finally, note that we proved the inevitability of late-time FLRW-solutions in which the kinetic and the potential energy evolve in the same parametric way with time.
This is in line with analogous results in simpler theories \cite{Shiu:2023nph, Shiu:2023fhb, Shiu:2023yzt}.
Such conclusions might be instrumental for formulations of the distance conjecture \cite{Ooguri:2006in} in the presence of a potential \cite{Debusschere:2024rmi}.

\onecolumngrid
\newpage

\begin{acknowledgments}
\subsection*{Acknowledgments}
We would like to thank Filippo Revello, Thomas Van Riet, and especially Muthusamy Rajaguru for helpful exchanges of ideas. GS is supported in part by the DOE grant DE-SC0017647. FT is supported by the FWO Odysseus grant GCD-D0989-G0F9516N. HVT is supported in part by the NSF CAREER grant DMS-1843320 and a Vilas Faculty Early-Career Investigator Award.
\end{acknowledgments}

\appendix

\section{Late-time cosmologies with scalars and axions} \label{app: late-time scalar-axion cosmologies}
In a $d$-dimensional FLRW-background, eqs.~(\ref{axion-scalar axion FRW-KG eq.}, \ref{axion-scalar scalar FRW-KG eq.},  \ref{axion-scalar Friedmann eq.}), can be expressed as a constrained system of first-order autonomous ordinary differential equations \cite{Halliwell:1986ja, Copeland:1997et, Coley:1999mj, Guo:2003eu, Cicoli:2020cfj}. Let
\begin{align*}
    x^a & = \dfrac{\kappa_d}{\sqrt{d-1} \sqrt{d-2}} \, \dfrac{\dot{\phi}^a}{H}, \\
    w^r & = \dfrac{\kappa_d}{\sqrt{d-1} \sqrt{d-2}} \, \dfrac{\dot{\zeta}^r}{H} \, \e^{-\frac{1}{2} \, \kappa_d \sum_{b} \lambda_{r b} \phi^b}, \\
    y^i & = \dfrac{\kappa_d \sqrt{2}}{\sqrt{d-1} \sqrt{d-2}} \, \dfrac{1}{H} \, \sqrt{\Lambda_i \, \e^{- \kappa_d \gamma_{i a} \phi^a}},
\end{align*}
and
\begin{align*}
    f & = (d-1) H, \\[1.00ex]
    c_{i a} & = \dfrac{1}{2} 
    \dfrac{\sqrt{d-2}}{\sqrt{d-1}} \, \gamma_{ia}, \\
    l_{r a} & = \dfrac{1}{2} 
    \dfrac{\sqrt{d-2}}{\sqrt{d-1}} \, \lambda_{ra}.
\end{align*}
Then, the cosmological equations read
\begin{subequations}
\begin{align}
    \dot{x}^a & = \biggl[ -x^a (y)^2 + \sum_{i=1}^m {c_i}^{a} (y^i)^2 - \sum_{r=1}^p {l_r}^a (w^r)^2 \biggr] \, f, \label{axion-scalar x-equation} \\
    \dot{w}^r & = \biggl[ - (y)^2 + \sum_{a=1}^n l_{ra} x^a \biggr] \, w^r f, \label{axion-scalar w-equation} \\
    \dot{y}^i & = \biggl[ 1 - (y)^2 - \sum_{a=1}^n c_{ia} x^a \biggr] \, y^i f, \label{axion-scalar y-equation}
\end{align}
\end{subequations}
jointly with the two conditions
\begin{subequations}
\begin{align}
    & \dot{f} = - [(x)^2 + (w)^2] f^2, \label{axion-scalar f-equation} \\
    & (x)^2 + (w)^2 + (y)^2 = 1. \label{axion-scalar sphere-condition}
\end{align}
\end{subequations}
Here, the index position is arbitrary: the $a$-indices relate to a field-space metric that is a Kronecker delta, while the $r$- and $i$-indices are just dummy labels; the shorthand notations are used $\smash{(x)^2 = \sum_{a=1}^m (x^a)^2}$, $\smash{(w)^2 = \sum_{r=1}^p (w^r)^2}$ and $\smash{(y)^2 = \sum_{i=1}^m (y^i)^2}$. Below are a series of mathematical results. A formulation in terms of physical observables and an analysis of their implications are in the main text.

\vspace{4pt}

Let the unknown functions be such that $x^a \in [-1, 1]$, $w^r \in [-1,1]$ and $y^i \in ]0,1]$; let $t_0$ be the initial time. Let $c_a = \min_i \, c_{ia}$ and $C_a = \max_i \, c_{ia}$. For brevity, let $\smash{\xi = (x)^2 + (w)^2}$.

\subsection{General preliminary results}

\begin{lemma} \label{lemma: axion-scalar f - lower bound}
    At all times $t > t_0$, one has
    \begin{equation} \label{eq.: axion-scalar f - lower bound}
        f(t) \geq \dfrac{1}{t-t_0 + \dfrac{1}{f(t_0)}}.
    \end{equation}
    In particular, this implies that $\smash{\int_{t_0}^\infty \de t \, f(t) = \infty}$.
\end{lemma}

\begin{proof}
Let $\smash{f_-}$ be a function such that $\smash{\dot{f}_- / (f_-)^2 = - 1}$. A simple integration gives $\smash{f_- = 1 / [t - t_0 + 1/f_-(t_0)]}$. In view of eq.~(\ref{axion-scalar sphere-condition}), eq.~(\ref{axion-scalar f-equation}) gives $\smash{\dot{f}/f^2 \geq \dot{f}_- / (f_-)^2}$; by integrating the latter, one immediately gets eq.~(\ref{eq.: axion-scalar f - lower bound}).
\end{proof}

\vspace{2pt}
\begin{lemma} \label{lemma: axion-scalar w,y>0}
    If $w^r(t_0)=0$, then $w^r(t)=0$ at all times $t>t_0$, and, similarly, if $y^i(t_0)=0$, then $y^i(t)=0$ at all times $t>t_0$.
\end{lemma}

\begin{proof}
    If $w^r(t_0)=0$, then $\dot{w}^r(t_0)=0$ by eq.~(\ref{axion-scalar w-equation}); equivalent statements hold for $y^i(t)$, coming from eq.~(\ref{axion-scalar y-equation}). Therefore, the conclusions hold.
\end{proof}

\vspace{4pt}
Let the function $\smash{\hat{\varphi}: \mathbb{R} \to \mathbb{R}_0^+}$ be defined as $\smash{\varphi(t) = \int_{t_0}^t \de s \; [y(s)]^2 \, f(s)}$. A thorough analysis of the dynamical system with $w^r=0$ is in refs.~\cite{Shiu:2023nph, Shiu:2023fhb}.

\vspace{2pt}
\begin{corollary}
    At all times $t>t_0$, one has
    \begin{equation}
        \dot{f} < 0.
    \end{equation}
\end{corollary}

\begin{proof}
    Because $[w(t)]^2>0$ by lemma \ref{lemma: axion-scalar w,y>0}, the conclusion holds by eq.~(\ref{axion-scalar f-equation}).
\end{proof}

\vspace{4pt}

As in ref.~\cite{Shiu:2023fhb}, let $s(t)$ be the function $s: \; [t_0, \infty[ \, \to \, [t_0, \infty[$ such that
\begin{equation}
    \dot{s}(t) = \dfrac{1}{f [s(t)]},
\end{equation}
with $s(t_0) = t_0$. In the notation in which every function $g=g(t)$ defines a function $\mathrm{g}(t) = g [s(t)]$, one can rearrange the dynamical system in eqs.~(\ref{axion-scalar x-equation}, \ref{axion-scalar w-equation}, \ref{axion-scalar y-equation}) and (\ref{axion-scalar f-equation}, \ref{axion-scalar sphere-condition}) by just getting rid of the function $f$. Also, note that $\smash{\int_{t_1}^{t_2} \de t \, \mathrm{g}(t) = \int_{s(t_1)}^{s(t_2)} \de s \, f(s) \, g(s)}$. As $\smash{\int_{t_0}^t \de r \; f[s(r)] \dot{s}(r) = \int_{t_0}^{s(t)} \de s \; f(s) = t}$, one has that $\smash{\lim_{t \to +\infty} s(t) = +\infty}$. Therefore, limits at infinity of $\mathrm{g}$-functions correspond to limits at infinity of $g$-functions. Let $\rmxi(t) = \xi(s(t))$.

\vspace{4pt}

Let $\mathrm{c}_i \in \mathbb{R}^n$ be the $m$ vectors with components $\smash{(\mathrm{c}_i)_a = c_{ia}}$ and let the associated convex hull be the parameter-space hypersurface $\mathrm{CH}(\lbrace \mathrm{c}_i \rbrace_{i=1}^m) = \bigl\lbrace \mathrm{q} \in \mathbb{R}^n: \; \mathrm{q}_a = \sum_{i=1}^m \! \sigma_i (\mathrm{c}_i)_a, \; (\sigma_i)_{i=1}^m \in (\mathbb{R}_0^+)^m, \; \sum_{i=1}^m \sigma_i = 1 \bigr\rbrace$. Similarly, one can define $\mathrm{l}_r \in \mathbb{R}^n$ to be the $p$ vectors with components $\smash{(\mathrm{l}_r)_a = l_{ra}}$. Given the infimum of the distance of the coupling convex hull from the origin, i.e.
\begin{equation} \label{convex-hull distance}
    c_\infty = \inf_{\mathrm{q} \in \mathrm{CH}} \sqrt{\mathrm{q}_a \mathrm{q}^a},
\end{equation}
let $\mathrm{c}_\infty$ be the vector joining the origin to the convex hull of length $c_\infty$. In this work, it is always assumed that $c_\infty >0$. The vectors $\mathrm{c}_i$ and $\mathrm{l}_r$ are referred to as the coupling and axion-coupling vectors, respectively. Let $\mathrm{P}(\lbrace \mathrm{c}_i \rbrace_{i=1}^m) \subset \mathrm{CH}(\lbrace \mathrm{c}_i \rbrace_{i=1}^m)$ be the hyperplane through the endpoint of the vector $\smash{\mathrm{c}_\infty}$ that is orthogonal to the vector $\smash{\mathrm{c}_\infty}$ itself. Up to relabeling indices, one has this general setup:
\begin{itemize}[leftmargin=*]
    \item a number $l \leq m$ of coupling vectors is on the hyperplane, i.e. $\smash{\mathrm{c}_{\infty} = \mathrm{c}_1, \mathrm{c}_{2}, \dots, \mathrm{c}_{l} \in\; \mathrm{P}}$, and the vector $\smash{\mathrm{c}_{\infty}}$ belongs to the convex hull of these $l$ vectors;
    \item all the remaining vectors have a distance from the origin larger than $c_\infty$, i.e. $\smash{(\mathrm{c}_j)^2 \geq c_\infty^2}$ for all indices $j = l+1, \dots, m$.
\end{itemize}
Pedagogical examples are e.g. in figs. \ref{fig.: axion-scalar antialignment}, \ref{fig.: general convex-hull criterion}, \ref{fig.: SUSY assisted inflation}, \ref{fig.: axion-scalar alignment}, \ref{fig.: large-gamma: late-time epsilon}.

\subsection{Universal convergence results}
In this section, analytic late-time bounds are proven for three classes of parameter-space configurations.

\subsubsection{Axion-scalar antialignment}
Here, all axion couplings -- i.e. for all $r$-indices -- are assumed to be such that
\begin{equation}
    \mathrm{c}_\infty \cdot \mathrm{l}_r \leq 0.
\end{equation}
A geometric picture of this setup is in fig. \ref{fig.: axion-scalar antialignment}.
Also, let $\smash{\rmvarphi(t) = \int_{t_0}^{t} \de s \; [\y(s)]^2}$ and, by assumption, let $\smash{\lim_{t \to \infty} \, \rmvarphi(t)=\infty}$.\footnote{In the presence of fluids with constant equation of state, this integral gets generalized in a way that always diverges \cite{axionscalarfluidcosmologies}.}

\vspace{4pt}
If $\smash{c_\infty^2 < 1}$, let $\smash{(\hx^a, \hy^\eta)}$ for $\eta=1, \dots, k$ be the critical points of the autonomous system such that $\smash{\check{\y}^\zeta = \check{\w}^r = 0}$ for $\zeta=k+1, \dots, m$ and $r=1,\dots,p$. In other words, $\smash{(\hx^a, \hy^\eta)}$ correspond to critical points of the autonomous system after the truncation of the potential to only the couplings in $\mathrm{P}$, for which case $\mathrm{rank} \, c_{\eta a} = k$; i.e. these are the critical points discussed in refs.~\cite{Collinucci:2004iw, Shiu:2023fhb}. In particular, note that
\begin{subequations}
    \begin{align}
        \sum_{a=1}^m c_{\eta a} \hx^a & = (\hx)^2, \label{eq.: cbarX-identity} \\
        \sum_{\eta=1}^k c_{\eta a} (\hy^\eta)^2 & = \hx_a (\hy)^2. \label{eq.: cbarY-identity}
    \end{align}
\end{subequations}
It is also apparent that $\smash{\hx = \mathrm{c}_\infty}$; consequently, one also has $\smash{(\hx)^2 = c_\infty^2}$.

\vspace{2pt}
\begin{theorem} \label{theorem: axion-scalar antialignment - X convergence}
    Let $y^\eta(t_0)>0$ for all indices $\eta=1,\ldots,k$. If $\smash{c_\infty^2 < 1}$, one has
    \begin{equation} \label{eq.: axion-scalar antialignment - X convergence}
        \lim_{t \to \infty} \, \x^a(t) = \hx^a.
    \end{equation}
\end{theorem}

\begin{proof}
    It is convenient to organize the proof in several major steps.
    \begin{enumerate}[label=\roman*)]
        \item \label{step i} By multiplying eq.~(\ref{axion-scalar x-equation}) by $\hx^a$ and summing over $a$, and taking advantage of the inequalities $\smash{\mathrm{c}_i \cdot \hx \geq (\hx)^2}$ and $\smash{\hx \cdot \mathrm{l}_r \leq 0}$, one can write the differential inequality
        \begin{align*}
            \hx \cdot \dot{\x} = - \hx \cdot \x \, (\y)^2 + \sum_{\eta=1}^m \mathrm{c}_\eta \cdot \hx \, (\y^\eta)^2 - \sum_{r=1}^p \mathrm{l}_r \cdot \hx \, (\w^r)^2 \geq - \bigl[ \hx \cdot \x - (\hx) ^2\bigr] \, (\y)^2.
        \end{align*}
        This integrates to the inequality $\smash{\hx \cdot \x(t) - (\hx)^2 \geq [\hx \cdot \x(t_0) - (\hx)^2] \, \e^{-\rmvarphi(t)}}$, which implies the asymptotic condition $\smash{\liminf_{t \to \infty} \, \hx \cdot \x(t) \geq (\hx)^2}$. This, in turn, implies the inequality
        \begin{align*}
            \liminf_{t \to \infty} \, \sqrt{[\x(t)]^2} \geq c_\infty.
        \end{align*}
        \item For the functions $\y^\eta$, for $\eta=1, \dots, k$, eq.~(\ref{axion-scalar y-equation}) can be written as $\smash{\dot{\y}^\eta/\y^\eta = (\x)^2 + (\w)^2 - \mathrm{c}_\eta \cdot \x}$, which can be combined into the equation
        \begin{align*}
            \sum_{\eta=1}^k (\hy^\eta)^2 \dfrac{\dot{\y}^\eta}{\y^\eta} = (\hy)^2 \bigl[ (\x)^2 + (\w)^2 \bigr] - \sum_{\eta=1}^k (\hy^\eta)^2 \, \mathrm{c}_i \cdot \x = (\hy)^2 \bigl[ (\x)^2 + (\w)^2 - \hx \cdot \x \bigr].
        \end{align*}
        After an integration, one gets the identity
        \begin{align*}
            \dfrac{1}{(\hy)^2} \, \ln \, \prod_{\eta=1}^k \biggl[ \dfrac{\y^\eta(t)}{\y^\eta(t_0)} \biggr]^{(\hy^\eta)^2} = \int_{t_0}^t \de s \; \bigl[ (\x(s))^2 + (\w(s))^2 - \hx \cdot \x(s) \bigr].
        \end{align*}
        Because $\smash{\y^\eta(t) > 0}$ at all times, there exists a positive real constant $\omega \in\; \mathbb{R}^+$, independent of time $t$, such that
        \begin{align*}
            \int_{t_0}^t \de s \; \bigl[ (\x(s))^2 + (\w(s))^2 - \hx \cdot \x(s) \bigr] \leq \omega.
        \end{align*}
        \item Because $\smash{(\x)^2/2 - (\hx)^2/2 \leq (\x)^2 - \hx \cdot \x}$, one can write the inequality
        \begin{align*}
            \int_{t_0}^t \de s \; \biggl[ \dfrac{1}{2} \, (\x(s))^2 + (\w(s))^2 - \dfrac{1}{2} \, (\hx)^2 \biggr] \leq \omega.
        \end{align*}
        Because $\smash{(\x)^2 + (\w)^2 = 1 - (\y)^2}$, an immediate consequence is the further inequality
        \begin{align*}
            \int_{t_0}^t \de s \; [\y(s)]^2 \geq \int_{t_0}^t \de s \; (1 - c_\infty^2) - 2 \omega = (1 - c_\infty^2) (t - t_0) - 2 \omega.
        \end{align*}
        This means that, for any positive real number $\smash{\delta \in\; ]0, 1-c_\infty^2[}$, there exists a time $t_\delta > t_0$ such that, at all times $t > t_\delta$, one has
        \begin{align*}
            \int_{t_0}^t \de s \; [\y(s)]^2 \geq \delta (t-t_0).
        \end{align*}
        By the results above, this means that one can write the inequality $\smash{(\hx)^2 - \hx \cdot \x(t) \leq \ab (\hx)^2 - \hx \cdot \x(t_0) \ab \, \e^{- \delta (t-t_0)}}$ for all times $t > t_\delta$. An integration then gives the chain of inequalities
        \begin{align*}
              \int_{t_0}^t \de s \; \bigl[ (\hx)^2 - \hx \cdot \x(s) \bigr] \leq \bigl\ab (\hx)^2 - \hx \cdot \x(t_0) \bigr\ab \int_{t_0}^t \de s \; \e^{- \delta s} \leq \psi_\delta,
        \end{align*}
        where $\psi_\delta$ is a real positive constant $\smash{\psi_\delta \in\; \ab (\hx)^2 - \hx \cdot \x(t_0) \ab/\delta, \infty[}$.
        \item Because $\smash{(\x - \hx)^2 = [(\x)^2 - \x \cdot \hx] + [(\hx)^2 - \x \cdot \hx]}$, one can combine the previous integral inequalities to get the integral inequality
        \begin{align*}
            \int_{t_0}^\infty \de s \; \bigl[ \x(t) - \hx \bigr] \!\cdot\! \bigl[ \x(t) - \hx \bigr] \leq \omega + \psi_\delta < \infty.
        \end{align*}
        By ref.~\cite[th. 2]{Shiu:2023fhb}, one concludes with eq.~(\ref{eq.: axion-scalar antialignment - X convergence}).
    \end{enumerate}
\end{proof}

\vspace{2pt}
\begin{corollary} \label{corollary: axion-scalar antialignment - W convergence}
    Let $y^\eta(t_0)>0$ for all indices $\eta=1,\ldots,k$. If $\smash{c_\infty^2 < 1}$ and $\smash{\mathrm{c}_\infty \cdot \mathrm{l}_r < 0}$, one has
    \begin{equation} \label{eq.: axion-scalar antialignment - W convergence}
        \lim_{t \to \infty} \, \w^r(t) = 0.
    \end{equation}
\end{corollary}

\begin{proof}
    Let $\smash{\mathrm{c}_\infty \cdot \mathrm{l}_r = - 2 u < 0}$. Because of eq.~(\ref{eq.: axion-scalar antialignment - X convergence}), one can write $\smash{\lim_{t \to \infty} \, \sum_{a=1}^n l_{r a} \x^a(t) = -2 u < 0}$. Hence, there exists a time $t_u > t_0$ such that, in view of eq.~(\ref{axion-scalar w-equation}), at all times $t > t_u$ the differential inequality holds
    \begin{align*}
        \dot{\w}^r = \biggl[ - (\y)^2 + \sum_{a=1}^n l_{r a} \x^a \biggr] \w^r \leq - u \w^r.
    \end{align*}
    An integration immediately gives eq.~(\ref{eq.: axion-scalar antialignment - W convergence}).
\end{proof}

\vspace{2pt}
\begin{corollary} \label{corollary: axion-scalar antialignment - xi convergence}
    Let $y^\eta(t_0)>0$ for all indices $\eta=1,\ldots,k$. If $\smash{c_\infty^2 < 1}$ and $\smash{\mathrm{c}_\infty \cdot \mathrm{l}_r < 0}$, one has
    \begin{equation} \label{eq.: axion-scalar antialignment - xi convergence}
        \lim_{t \to \infty} \, \emph{$\rmxi$}(t) = c_\infty^2.
    \end{equation}
\end{corollary}

\begin{proof}
    This is a trivial consequence of eqs.~(\ref{eq.: axion-scalar antialignment - X convergence}, \ref{eq.: axion-scalar antialignment - W convergence}).
\end{proof}

\vspace{2pt}
\begin{corollary} \label{corollary: axion-scalar antialignment for large c - xi convergence}
    If $c_\infty^2 \geq 1$, then $\lim_{t \to \infty} \, \xi(t) = 1$.
\end{corollary}

\begin{proof}
    If $c_\infty^2 > 1$, then by step \ref{step i} in the proof of theorem \ref{theorem: axion-scalar antialignment - X convergence}, one has $\lim_{t\to \infty} \, \rmvarphi(t)<+\infty$. This implies that $\lim_{t\to\infty} \, (\y(t))^2=0$, which gives the result.
    If $c_\infty^2 = 1$, one of two scenarios is in place. If $\lim_{t\to \infty}\varphi(t)<+\infty$, then the result follows by the above argument. If $\lim_{t\to \infty} \, \rmvarphi(t)=\infty$, then by step \ref{step i} in the proof of theorem \ref{theorem: axion-scalar antialignment - X convergence}, one has $\lim_{t\to\infty} \, [\x(t)]^2=1$, and the conclusion follows.
\end{proof}

\vspace{2pt}
\begin{remark} \label{remark: paper-2 generalization}
    Of course, eq.~(\ref{eq.: axion-scalar antialignment - X convergence}) also holds in the absence of axion couplings: this means that theorem \ref{theorem: axion-scalar antialignment - X convergence} is a generalization of the results of ref.~\cite[th. 2]{Shiu:2023fhb}. While the proof in ref.~\cite[th. 2]{Shiu:2023fhb} assumes that $\smash{\mathrm{rank} \, c_{ia} = m}$ and that the $(\y^i)^2$-terms are all non-negative at the proper critical point, the proof of theorem \ref{theorem: axion-scalar antialignment - X convergence} works for an arbitrary $c_{ia}$-matrix rank and arbitrary signs of the $\smash{(\y^i)^2}$-terms at any critical point. An example is in fig. \ref{fig.: general convex-hull criterion}.
\end{remark}

\subsubsection{Axion-scalar alignment}
Here, it is assumed that $\mathrm{rank} \, c_{ia} = \mathrm{rank} \, l_{ra} = n$, with $m, p \geq n$. This means that the matrices $c_{ia}$ and $l_{ra}$ are left-invertible, i.e. the matrices $\smash{(c^{-1})^{ai}}$ and $\smash{(l^{-1})^{ar}}$ exist such that $\smash{\sum_{i=1}^m (c^{-1})^{bi} c_{ia} = \delta_a^b}$ and $\smash{\sum_{r=1}^p (l^{-1})^{br} l_{ra} = \delta_a^b}$. Then, let $\smash{q^r = \min_{i} \, \sum_{a=1}^n c_{ia} (l^{-1})^{ar}}$ and $\smash{Q^r = \max_{i} \, \sum_{a=1}^n c_{ia} (l^{-1})^{ar}}$. 

\vspace{2pt}
\begin{lemma} \label{lemma: axion-scalar alignment - x upper bound}
    Let $q^r \geq 0$ for all $r$-indices. Let $w^r(t_0), y^i(t_0)>0$ for all $r$- and $i$-indices.
    Then one has
    \begin{equation} \label{eq.: axion-scalar alignment - x upper bound}
        \limsup_{t \to \infty} \, \sum_{a=1}^n \x_a(t) (l^{-1})^{ar} \leq Q^r.
    \end{equation}
\end{lemma}

\begin{proof}
    Formal solutions to eqs.~(\ref{axion-scalar w-equation}, \ref{axion-scalar y-equation}) can be written as
    \begin{align*}
        \ln \, \biggl[ \dfrac{\w^r(t)}{\w^r(t_0)} \biggr] & = - \rmvarphi(t) + \sum_{a=1}^n l_{ra} \int_{t_0}^t \de s \; \x^a(s), \\
        \ln \, \biggl[ \dfrac{\y^i(t)}{\y^i(t_0)} \biggr] & = - \rmvarphi(t) - \sum_{a=1}^n c_{ia} \int_{t_0}^t \de s \; \x^a(s) + (t-t_0). 
    \end{align*}
    A simple linear combination shows that
    \begin{align*}
        \sum_{r=1}^p (l^{-1})^{ar} \, \ln \, \biggl[ \dfrac{\w^r(t)}{\w^r(t_0)} \biggr] + \sum_{i=1}^m (c^{-1})^{ai} \, \ln \, \biggl[ \dfrac{\y^i(t)}{\y^i(t_0)} \biggr] = - \rmvarphi(t) \biggl[  \sum_{r=1}^p (l^{-1})^{ar} + \sum_{i=1}^m (c^{-1})^{ai} \biggr] + (t-t_0) \, \sum_{i=1}^m (c^{-1})^{ai},
    \end{align*}
    which can be further manipulated to show the identity
    \begin{align*}
        \sum_{a=1}^n \sum_{r=1}^p c_{ia} (l^{-1})^{ar} \, \ln \, \biggl[ \dfrac{\w^r(t)}{\w^r(t_0)} \biggr] + \ln \, \biggl[ \dfrac{\y^i(t)}{\y^i(t_0)} \biggr] = - \rmvarphi(t) \biggl[ 1 + \sum_{a=1}^n \sum_{r=1}^p c_{ia} (l^{-1})^{ar} \biggr] + (t-t_0).
    \end{align*}
    By assumption, one has $\smash{\sum_{a=1}^n c_{ia} (l^{-1})^{ar} \geq q^r \geq 0}$ for all $i$- and $r$-indices. At the same time, in view of eq.~(\ref{axion-scalar sphere-condition}), one has $\w^r(t), \y^i(t) \in\; ]0,1]$ at all times $t \geq t_0$ for all $i$- and $r$-indices. Therefore, there exists a positive real constant $\zeta \in \mathbb{R}^+$ such that the inequality holds
    \begin{align*}
        \rmvarphi(t) \geq \dfrac{t - t_0}{\displaystyle 1 + \sum_{a=1}^n \sum_{r=1}^p c_{ia} (l^{-1})^{ar}} - \zeta.
    \end{align*}
    In particular, this implies that $\smash{\lim_{t \to \infty} \, \rmvarphi(t) = \infty}$. In view of eq.~(\ref{axion-scalar x-equation}), one can write the differential inequality
    \begin{align*}
        \sum_{a=1}^n \dot{\x}_a (l^{-1})^{ar} = - \sum_{a=1}^n \x_a (l^{-1})^{ar} \, (\y)^2 + \sum_{a=1}^n \sum_{i=1}^m c_{i a} (l^{-1})^{ar} (\y^i)^2 - (\w)^2 \leq - \biggl[ \sum_{a=1}^n \x_a (l^{-1})^{ar} - Q^r \biggr] \, (\y)^2.
    \end{align*}
    An integration leads to the inequality
    \begin{align*}
        \sum_{a=1}^n \x_a(t) (l^{-1})^{ar} \leq Q^r + \biggl[ \sum_{a=1}^n \x_a(t_0) (l^{-1})^{ar} - Q^r\biggr] \, \e^{- \rmvarphi(t)},
    \end{align*}
    which, then, because $\smash{\lim_{t \to \infty} \, \rmvarphi(t) = \infty}$, immediately implies eq.~(\ref{eq.: axion-scalar alignment - x upper bound}).
\end{proof}

\vspace{4pt}
Let $\smash{\tilde{C}_a = \sum_{r=1}^p Q^r l_{ra}}$. By assumption, let $\smash{c_{ia}, l_{ra} \geq 0}$ for all $i$-, $r$- and $a$-indices. 
An example is in fig. \ref{fig.: axion-scalar alignment}.

\vspace{2pt}
\begin{lemma} \label{lemma: axion-scalar alignment - w=0 and x>c}
    Let $w^r(t_0), y^i(t_0)>0$ for all $r$- and $i$-indices. If $\smash{\sum_{a=1}^n (c_{ia} + l_{ra}) \tilde{C}^a < 1}$ for all $i$-indices, then one has
    \begin{equation} \label{eq.: axion-scalar alignment - w=0}
        \lim_{t \to \infty} \, \w^r(t) = 0.
    \end{equation}
    Moreover, if the condition $\smash{\sum_{a=1}^n (c_{ia} + l_{ra}) \tilde{C}^a < 1}$ holds for all $r$-indices, one has
    \begin{equation} \label{eq.: axion-scalar alignment - x>c}
        \liminf_{t \to \infty} \, \x^a(t) \geq c^a.
    \end{equation}
\end{lemma}

\begin{proof}
    A linear combination of the formal solutions to eqs.~(\ref{axion-scalar w-equation}, \ref{axion-scalar y-equation}) allows one to write the identity
    \begin{align*}
        \ln \, \biggl[ \dfrac{\y^i(t)}{\y^i(t_0)} \biggr] - \ln \, \biggl[ \dfrac{\w^r(t)}{\w^r(t_0)} \biggr] = (t - t_0) - \sum_{a=1}^n (c_{ia} + l_{ra}) \int_{t_0}^t \de s \; \x^a(s). 
    \end{align*}
    In view of eq.~(\ref{eq.: axion-scalar alignment - x upper bound}), for any arbitrary positive real constant $\delta \in \mathbb{R}^+$, there exists a time $\smash{t_\delta > t_0}$ such that the identity above implies the inequality
    \begin{align*}
        \ln \, \biggl[ \dfrac{\y^i(t)}{\y^i(t_0)} \biggr] - \ln \, \biggl[ \dfrac{\w^r(t)}{\w^r(t_0)} \biggr] \geq (t - t_0) \, \biggl[ 1 - \sum_{a=1}^n (c_{ia} + l_{ra}) \tilde{C}^a \biggr].
    \end{align*}
    Because $\sum_{a=1}^n (c_{ia} + l_{ra}) \tilde{C}^a < 1$, given $\smash{\delta = [1 - \sum_{a=1}^n (c_{ia} + l_{ra}) \tilde{C}^a]/2}>0$, one can write the inequality
    \begin{align*}
        \dfrac{\y^i(t)}{\w^r(t)} \geq \dfrac{\y^i(t_0)}{\w^r(t_0)} \, \e^{\delta \, (t - t_0)}.
    \end{align*}
    Because $\w^r(t), \y^i(t) \in\; ]0,1]$ at all times $t \geq t_0$ for all $i$- and $r$-indices, eq.~(\ref{eq.: axion-scalar alignment - w=0}) follows immediately. In case this holds for all $r$-indices, then there exists a time $t_\infty > t_\delta$ such that eq.~(\ref{axion-scalar x-equation}) implies the chain of differential inequalities
    \begin{align*}
        \dot{\x}_a \geq - \x_a (\y)^2 + \sum_{i=1}^m c_{ia} (\y^i)^2 \geq - (\x_a - c_a) \, (\y)^2,
    \end{align*}
    which implies eq.~(\ref{eq.: axion-scalar alignment - x>c}).
\end{proof}

\vspace{2pt}
\begin{theorem} \label{theorem: axion-scalar alignment - xi bounds}
    Let $w^r(t_0), y^i(t_0)>0$ for all $r$- and $i$-indices. If the condition $\smash{\sum_{a=1}^n (c_{ia} + l_{ra}) \tilde{C}^a < 1}$ holds for all $i$- and $r$-indices, and if $l^a, c^a \geq 0$ for all $a$-indices, then one has
    \begin{subequations}
        \begin{align}
            \liminf_{t \to \infty} \, \emph{\rmxi}(t) & \geq (c)^2, \label{eq.: axion-scalar alignment - xi lower bound} \\
            \limsup_{t \to \infty} \, \emph{\rmxi}(t) & \leq (\tilde{C})^2. \label{eq.: axion-scalar alignment - xi upper bound}
        \end{align}
    \end{subequations}
    Of course, if $\tilde{C}_a = c_a$, one has
    \begin{equation}
        \lim_{t \to \infty} \, \emph{\rmxi}(t) = (c)^2.
    \end{equation}
\end{theorem}

\begin{proof}
    The proof follows immediately from lemmas \ref{lemma: axion-scalar alignment - x upper bound} and \ref{lemma: axion-scalar alignment - w=0 and x>c}.
\end{proof}

\vspace{2pt}
\begin{theorem} \label{theorem: axion-scalar alignment - xi general bound}
    Let $w^r(t_0), y^i(t_0)>0$ for all $r$- and $i$-indices. If $\smash{c_{ia} \geq 0}$ for all $i$- and $a$-indices, and if $\sum_{a=1}^n C_a \tilde{C}^a\leq 1$, then one has
    \begin{equation} \label{eq.: axion-scalar alignment - xi general bound}
        \limsup_{t \to \infty} \, \emph{\rmxi}(t) \leq \sum_{a=1}^n C_a \tilde{C}^a.
    \end{equation}
\end{theorem}

\begin{proof}
    As $l_{ra} \geq 0$, eq.~(\ref{eq.: axion-scalar alignment - x upper bound}) implies the bound $\smash{\limsup_{t \to \infty} \, \x_a(t) \leq \tilde{C}_a}$. Then, because $C_a \geq c_{ia} \geq 0$, eq.~(\ref{axion-scalar y-equation}) means there exists a time $t_\infty$ such that the differential inequality holds
    \begin{align*}
        \dot{\y}^i \geq \biggl[ 1 - (\y)^2 - \sum_{a=1}^n C_a \tilde{C}^a \biggr] \y^i.
    \end{align*}
    After defining the function $\smash{\zeta = \sum_{i=1}^n (\y^i)^2}$ for brevity, one can immediately obtain the further differential inequality $\smash{\dot{\zeta} \geq 2 \zeta \, [1 - \zeta - \sum_{a=1}^n C_a \tilde{C}^a]}$, which is valid at any time $t_\star > t_\infty$. Then, the latter can be integrated to give the late-time inequality $\smash{\zeta(t) \geq \zeta(t_\star) (1 - \sum_{a=1}^n C_a \tilde{C}^a) / [\zeta(t_\star) + \e^{-2 (1 - \sum_{b} C_b \tilde{C}^b) \, (t - t_\star)} \, (1 - \sum_{b=1}^n C_b \tilde{C}^b - \zeta(t_\star))]}$, i.e.
    \begin{align*}
        [\y(t)]^2 \geq \dfrac{\displaystyle [\y(t_\star)]^2 \, \biggl[ 1 - \sum_{a=1}^n C_a \tilde{C}^a \biggr]}{\displaystyle [\y(t_\star)]^2 + \e^{-2 (1 - \sum_{b} C_b \tilde{C}^b) \, (t - t_\star)} \biggl[ 1 - \sum_{b=1}^n C_b \tilde{C}^b - [\y(t_\star)]^2 \biggr]}.
    \end{align*}
    In view of eqs.~(\ref{axion-scalar f-equation}, \ref{axion-scalar sphere-condition}), eq.~(\ref{eq.: axion-scalar alignment - xi general bound}) follows straightforwardly.
\end{proof}

\subsection{Partial coupling misalignment}
Let $\smash{L_a = \max_r \, l_{ra}}$.

\vspace{4pt}
\begin{lemma} \label{lemma: axion-scalar x - lower bound for L<0 and c+L<1}
    Let $\smash{L^a \leq 0}$ and $\smash{c^a + L^a \leq 1}$. If $\smash{\lim_{t \to \infty} \emph{{\rmvarphi}}(t) = \infty}$, then one has
    \begin{equation} \label{eq.: axion-scalar x - lower bound for L<0 and c+L<1}
        \liminf_{t \to \infty} \x^a(t) \geq c^a + L^a.
    \end{equation}
\end{lemma}

\begin{proof}
    In view of eq.~(\ref{axion-scalar x-equation}), one can write the inequality
    \begin{align*}
        \dot{\x}^a = -\x^a (\y)^2 + \sum_{i=1}^m {c_{i}}^{a} (\y^i)^2 - \sum_{r=1}^p l_{r a} (\w^r)^2 \geq [-\x^a + c^a] \, (\y)^2 - L^a \, (\w)^2.
    \end{align*}
    If $L^a \leq 0$, then eq.~(\ref{axion-scalar sphere-condition}) allows one to write the further chain of inequalities
    \begin{align*}
        \dot{\x}^a \geq [-\x^a + c^a + L^a] \, (\y)^2 - L^a \, [1 - (\x)^2] \geq [-\x^a + c^a + L^a] \, (\y)^2.
    \end{align*}
    By means of a simple integration, one then finds
    \begin{align*}
        \x^a(t) \geq (c^a + L^a) + \e^{-\rmvarphi(t)} \bigl[ \x^a(t_0) - (c^a + L^a) \bigr].
    \end{align*}
    If $c^a + L^a \leq 1$, then eq.~(\ref{eq.: axion-scalar x - lower bound for L<0 and c+L<1}) follows straightforwardly.
\end{proof}

\vspace{4pt}
\begin{corollary} \label{corollary: finite discriminant for L<0 and c+L>1}
    If $L^a \leq 0$ and $c^a + L^a > 1$ for at least one $a$-index, then $\smash{\lim_{t \to \infty} \emph{{\rmvarphi}}(t) < \infty}$.
\end{corollary}

\begin{proof}
    If $L^a \leq 0$ and $c^a + L^a > 1$, then eq.~(\ref{eq.: axion-scalar x - lower bound for L<0 and c+L<1}) cannot hold by eq.~(\ref{axion-scalar sphere-condition}), which forces the inequality $(\x)^2 \leq 1$ and therefore the inequality $(\x^a)^2 \leq 1$. Therefore, one cannot have $\smash{\lim_{t \to \infty} \rmvarphi(t) = \infty}$.
\end{proof}

\vspace{4pt}
\begin{corollary} \label{corollary: xi-limit for L<0 and c+L>1}
    If $L^a \leq 0$ and $c^a + L^a \geq 1$ for at least one $a$-index, then
    \begin{equation} \label{eq.: xi-limit for L<0 and c+L>1}
        \lim_{t \to \infty} \, \xi(t) = 1.
    \end{equation}
\end{corollary}

\begin{proof}
    By corollary \ref{corollary: finite discriminant for L<0 and c+L>1}, if $c^a + L^a > 1$, then $\smash{\lim_{t \to \infty} \rmvarphi(t) < \infty}$. By definition, then one must have $\smash{\lim_{t \to \infty} \, [\y(t)]^2 = 0}$. Hence, eq.~(\ref{eq.: xi-limit for L<0 and c+L>1}) follows immediately.
    If $c^a + L^a = 1$, one of two scenarios is in place. If $\smash{\lim_{t \to \infty} \rmvarphi(t) < \infty}$, the proof follows as above. If $\smash{\lim_{t \to \infty} \rmvarphi(t) = \infty}$, one has $\smash{\lim_{t \to \infty} \, \x^a(t) = 1}$ by eq.~(\ref{eq.: axion-scalar x - lower bound for L<0 and c+L<1}), which implies eq.~(\ref{eq.: xi-limit for L<0 and c+L>1}).
\end{proof}

\subsection{Critical points} \label{subapp: critical points}
Although the analytic results in this work are independent of the dynamical system reaching one of its critical points, for ease of presentation it is convenient to list the critical points of eqs.~(\ref{axion-scalar x-equation}, \ref{axion-scalar w-equation}, \ref{axion-scalar y-equation}) and (\ref{axion-scalar f-equation}, \ref{axion-scalar sphere-condition}) here. It is particularly instructive to compare the values of the parameters $\smash{\xi = 1/k}$ for the various solutions.

The autonomous system in eqs.~(\ref{axion-scalar x-equation}, \ref{axion-scalar w-equation}, \ref{axion-scalar y-equation}) and (\ref{axion-scalar f-equation}, \ref{axion-scalar sphere-condition}) has several critical points. Of course, three general classes of solutions are the ones where $\smash{(\bx^a,\by^i)}$, $\smash{(\hx^a, (\hy^\eta, \check{\y}^\zeta=0))}$, and $\smash{(\tilde{\x}^a, \tilde{\y}^i=0)}$, respectively, are the proper critical points, non-proper and degenerate non-proper critical points for the canonical-scalar-only theory, respectively \cite{Collinucci:2004iw, Shiu:2023fhb}, with $\smash{\bw^r=0}$, $\smash{\hw^r=0}$, and $\smash{\tilde{\w}^r=0}$, respectively. New solutions generally exist, too. These are discussed below.

\vspace{2pt}
\begin{lemma} \label{lemma: axion-scalar critical points}
    For each index $\smash{i_*=1,\dots,m}$, the constrained autonomous system in eqs.~(\ref{axion-scalar x-equation}, \ref{axion-scalar w-equation}, \ref{axion-scalar y-equation}) and (\ref{axion-scalar f-equation}, \ref{axion-scalar sphere-condition}) admits the critical points
    \begin{subequations}
    \begin{align}
        \x^a_* & = (\y^{i_*}_*)^2 \sum_{r=1}^p (l^{-1})^{ar}, \label{eq.: axion-scalar x critical points} \\
        (\w^r_*)^2 & = -(\y^{i_*}_*)^2 \sum_{a=1}^n (l^{-1})^{ar} (x_{* a} - c_{i_*a}), \label{eq.: axion-scalar w critical points} \\
        (\y^{i_*}_*)^2 & = \dfrac{1}{\displaystyle 1 + \sum_{a=1}^n \sum_{r=1}^p c_{i_*a} (l^{-1})^{ar}}. \label{eq.: axion-scalar y critical points}
    \end{align}
    \end{subequations}
    In particular, the $k$-parameter reads
    \begin{equation} \label{eq.: axion-scalar xi critical points}
        \dfrac{1}{k_*} = \dfrac{\displaystyle \sum_{b=1}^n \sum_{s=1}^p c_{i_*b} (l^{-1})^{bs}}{\displaystyle 1 + \sum_{a=1}^n \sum_{r=1}^p c_{i_*a} (l^{-1})^{ar}}.
    \end{equation}
\end{lemma}

\begin{proof}
    One finds immediately that it should be $\smash{(\y_*)^2 = 1 / \bigl[ 1 + \sum_{a} \sum_r c_{ia} (l^{-1})^{ar} \bigr]}$, assuming the left-inverse matrix $\smash{(l^{-1})^{ar}}$ to exist. However, this is not a generic solution, since it has an index-dependence: it can possibly be a solution only if all the vector components $\smash{v_i = \sum_{a} \sum_r c_{ia} (l^{-1})^{ar}}$ happen to be identical. A possible solution is however the one where $\smash{(\y^{i_*}_*)^2=(\y_*)^2}$, with $\smash{(\y^{i'}_*)=0}$ for $i'=1,\dots,i_*-1,i_*+1,\dots,m$, for a specific index $i_*$. In this case, the solutions exist and take the form in eqs.~(\ref{eq.: axion-scalar x critical points}, \ref{eq.: axion-scalar w critical points}, \ref{eq.: axion-scalar y critical points}) and (\ref{eq.: axion-scalar xi critical points}).
\end{proof}

\vspace{2pt}
\begin{remark}
    In principle, there may also exist additional critical points, where $\smash{\tilde{\y}_*^i=0}$ for each $i=1,\dots,m$ and $\smash{\tilde{\w}_*^r \neq 0}$ for at least some $r$-indices (the case where $\smash{\tilde{\w}_*^r = 0}$ for all $r$-indices has already been listed above). These only exist if both conditions $\smash{\sum_{r=1}^p l_{ra} (\tilde{w}_*^r)^2 = 0}$ and $\smash{\sum_{a=1}^n l_{ra} \tilde{x}_*^a = 0}$ have solutions. This requires $\smash{\mathrm{rank} \, l_{ra} < p}$.
\end{remark}

\vspace{2pt}
\begin{remark}
    In the minimal non-trivial case, with $n=p=m=1$, one has $\smash{\x_* = \y_*^2 / l = 1/(c+l)}$, $\smash{\y_*^2 = l/(c+l)}$ and $\smash{\w_*^2 = [c(c+l) - 1]/(c+l)^2}$: these are the non-geodesic solutions of ref.~\cite{Cicoli:2020cfj}. In this case, one finds $\smash{1/k_* = c/(c+l)}$; note that the axion-free proper critical-point solution has $\smash{1/k = c^2}$, which means that $\smash{1/k_* \leq 1/k}$ if $l \geq (1-c^2)/c$.
\end{remark}

\section{Axion couplings in string-theoretic models} \label{app.: axion couplings in string-theoretic models}

In this subsection, we overview the scalings of several axion couplings that, together with the potentials summarized in ref.~\cite[app. B]{Shiu:2023fhb}, appear in a plethora of string-theoretic models.

\subsection{Example: 4-dimensional supersymmetric compactifications}

In type-II 4-dimensional compactifications with $N_4=1$ supersymmetry, in certain regions of the field space, one can find Kähler potentials of the form $\kappa_4^2 K = - \sum_{a} n^a \, \mathrm{ln} \, [-\I (\xi^a - \overline{\xi}{}^a)]$ for a set of chiral multiplets $\smash{\xi^a = \theta^a + \I \, \e^{l^a \varphi^a}}$. Here, $l^a$ and $n^a > 0$ are real constants that depend on the microscopic origin of the fields, and the scalars $\varphi^a$ approach the boundary as $l^a \varphi^a \to \infty$. Then, the kinetic energy
\begin{align*}
    T[\theta, \varphi] = \sum_{a} \dfrac{n^a}{4 \kappa_4^2} \, \bigl[ \e^{-2 l^a \varphi^a} (\der \theta^a)^2 + l^2 (\der \varphi^a)^2 \bigr]
\end{align*}
can be recast into the form of eq.~(\ref{axion-scalar kinetic energy}), with $\smash{\lambda_{r a} = (2\sqrt{2}/\sqrt{n^a}) \, \delta_{r a}}$, through the definitions
\begin{align*}
    \phi^a & = \dfrac{l^a}{\kappa_4} \dfrac{\sqrt{n^a}}{\sqrt{2}} \, \varphi^a, \\
    \zeta^a & = \dfrac{1}{\kappa_4} \dfrac{\sqrt{n^a}}{\sqrt{2}} \, \theta^a.
\end{align*}

For instance, consider the type-IIB axio-dilaton $\tau = C_0 + \I \, \e^{-\sdil}$ and Kähler modulus $\rho = \alpha + \I \, \e^{4 \omega}$, where $\sdil$ is the 10-dimensional dilaton, $\omega$ is the non-canonical Einstein-frame radion -- with $\smash{\mathrm{vol}_{\mathrm{CY}_3} = \e^{6 \omega} \, l_s^6}$ --, and $C_0$ and $\alpha$ are the 0- and 4-RR-form axions, respectively. In a 4-dimensional Calabi-Yau orientifold compactification, their purely kinetic action can be read off the Kähler potential \cite{Giddings:2001yu, Grimm:2004uq}
\begin{align*}
    \kappa_4^2 K = - \mathrm{ln} \, [-\I (\tau - \overline{\tau})] - 3 \, \mathrm{ln} \, [-\I (\rho - \overline{\rho})] + \mathrm{ln} \, \dfrac{2}{\pi}.
\end{align*}
In this way, the kinetic action reads
\begin{align*}
    S[\tilde{C}_0, \tsdil; \tilde{\alpha}, \tomega] = \int_{\mathrm{X}_{1,3}} \de^{1,3} x \sqrt{- g_{1,3}} \; \Biggl[ \dfrac{1}{2} \, \e^{2 \sqrt{2} \, \kappa_4 \tsdil} (\der \tilde{C}_0)^2 + \dfrac{1}{2} \, (\der \tsdil)^2 + \dfrac{1}{2} \, \e^{- \frac{2}{3} \sqrt{6} \, \kappa_4 \tomega} (\der \tilde{\alpha})^2 + \dfrac{1}{2} \, (\der \tomega)^2 \Biggr],
\end{align*}
where $\smash{\kappa_4 \tsdil = \sdil/\sqrt{2}}$, $\smash{\kappa_4 \tomega = 2 \sqrt{6} \, \omega}$, $\smash{\kappa_4 \tilde{C}_0 = C_0/\sqrt{2}}$ and $\smash{\kappa_4 \tilde{\alpha} = \sqrt{3/2} \, \alpha}$. In the rotated basis where one works in terms of the canonical 4-dimensional dilaton $\smash{\tdelta}$ and string-frame radion $\smash{\tsigma}$, one finds
\begin{align*}
    S[\tilde{C}_0, \tdelta; \tilde{\alpha}, \tsigma] = \int_{\mathrm{X}_{1,3}} \de^{1,3} x \sqrt{- g_{1,3}} \; \Biggl[ \dfrac{1}{2} \, \e^{\sqrt{2} \, \kappa_4 \tdelta + \sqrt{6} \, \kappa_4 \tsigma} (\der \tilde{C}_0)^2 + \dfrac{1}{2} \, (\der \tdelta)^2 + \dfrac{1}{2} \, \e^{\sqrt{2} \kappa_4 \tdelta - \frac{\sqrt{6}}{3} \kappa_4 \tsigma} (\der \tilde{\alpha})^2 + \dfrac{1}{2} \, (\der \tsigma)^2 \Biggr].
\end{align*}
As for the scalar potential, according to refs.~\cite{Shiu:2023nph, Shiu:2023fhb}, any possible perturbative 4-dimensional dilaton coupling is such that $\smash{\gamma_{\tdelta} \leq - \sqrt{2}}$, whereas couplings for the string-frame radion are model-dependent.

As another example, through asymptotic Hodge theory, one can compute the Kähler potentials for the complex-structure moduli appearing F-theory compactifications on Calabi-Yau fourfolds $\mathrm{Y}_4$ with 4-form flux. One has several complex multiplets $t^a = \theta^a + \I \, s^a$, with $a = 1,\dots,h=h^{3,1}(\mathrm{Y}_4)$. Here, $s^a$ are F-theory complex-structure moduli (saxions), and $\theta^a$ are the corresponding axions. In a given growth sector, i.e. a given region of the $s^a$-field moduli space, the Kähler potential takes the form \cite{Grimm:2019ixq}
\begin{align*}
    \kappa_4^2 K & = - \mathrm{ln} \, \biggl[ \biggl( \dfrac{s^1}{s^2} \biggr)^{\! d_1} \biggl( \dfrac{s^2}{s^3} \biggr)^{\! d_2} \dots (s^h)^{d_h} \, f(\upsilon, \overline{\upsilon}) \biggr] = -\sum_{a=1}^h l_a \, \mathrm{ln} \, \Bigl[ - \dfrac{\I}{2} (t^a - \overline{t}^a) \Bigr] - \mathrm{ln} \, f(\upsilon, \overline{\upsilon}),
\end{align*}
where the function $\smash{f(\upsilon, \overline{\upsilon})}$ involves the scalars remaining finite, and can thus be neglected asymptotically, and $d_i$ are constants such that $d_{a} \geq d_{a-1} \geq 0$, for $a=1,\dots,h$, with an extra formal term $d_0=0$. Also, we have $l_i = d_{a} - d_{a-1}$. The fact that the dynamics leave the theory in a given growth sector is an assumption. Examples of scalar potentials for these theories can be found throughout refs.~\cite{Grimm:2019ixq, Calderon-Infante:2022nxb}, though the Einstein-frame volume dependence must also be added. In our notation, we define the canonically-normalized scalars as $\smash{\kappa_4 \phi^a = \sqrt{l_i/2} \, \mathrm{ln} \, s^a}$ and the rescaled axions as $\smash{\kappa_4 \zeta^a = \sqrt{l_a/2} \, \theta^a}$. In all these cases, the axion couplings are $\smash{\lambda_{a b} = 2 \sqrt{2 l_a} \, \delta_{a b}}$.

\subsection{RR-, NSNS- and YM-axions}

Here we compute the couplings for the axions that emerge from higher-dimensional RR-, NSNS-, and YM-fields.

\begin{itemize}
    \item Given the $(q-1)$-RR-form $C_{q-1}$ in a 10-dimensional type-II superstring theory, an associated axion $\theta_{\R\R^q}$ in the $d$-dimensional compactified theory comes from the Kaluza-Klein decomposition $C_{q-1} = \theta_{\R\R^q} \, \breve{\alpha}_{q-1} l_s^{q-1}$. Here, $\theta_{\R\R^q}$ is a $d$-dimensional field and $\smash{\breve{\alpha}_{q-1}}$ is a harmonic $(q-1)$-form in the internal space $\smash{\mathrm{K}_{10-d}}$, with $q-1 \leq 10-d$. The presence of the axion $\theta_{\R\R^q}$ means that the total $q$-form field-strength tensor reads $\smash{F_q = \de \theta_{\R\R^q} \wedge \breve{\alpha}_{q-1} l_s^{q-1} + f^{(q)} \breve{\alpha}_q l_s^{q-1}}$, where $\smash{f^{(q)} \in \mathbb{Z}}$ is a flux-quantization integer and $\smash{\breve{\alpha}_q}$ is a harmonic $q$-form in the internal space $\smash{\mathrm{K}_{10-d}}$. The $(q-1)$-RR-form appears in the string-frame action in the term
    \begin{align*}
        S[F_q] & = \dfrac{1}{2 \kappa_{10}^2} \int_{\mathrm{X}_{1,9}} \biggl[ - \dfrac{1}{2} \, F_q \wedge \star_{1,9} F_q \biggr].
    \end{align*}
    Given the metric ansatz in ref.~\cite[eq.~(B.2)]{Shiu:2023fhb}, the dimensional reduction of the flux term leads to the scalar potential in ref.~\cite[eq.~(B.14)]{Shiu:2023fhb}. One is left with the dimensional reduction of the term
    \begin{align*}
        S[\theta_{\R\R^q}] = \dfrac{1}{2 \kappa_{d}^2} \, g_s^2 \ab \breve{a}_{q-1} \ab^2 \int_{\mathrm{X}_{1,d-1}} \de^{1,d-1} x \, \sqrt{- \tilde{g}_{1,d-1}} \; \biggl[ - \dfrac{1}{2} \, \der_\mu \theta_{\R\R^q} \der_\nu \theta_{\R\R^q} \, \tilde{g}^{\mu \nu} \,  \e^{2 \delta + [(10-d) - 2 (q-1)] \sigma} \biggr],
    \end{align*}
    where we introduced the notation $\smash{l_s^{10-d} \ab \breve{a}_{q-1} \ab^2 = \int_{\mathrm{K}_{10-d}} \breve{a}_{q-1} \wedge \breve{*}_{10-d} \, \breve{a}_{q-1}}$, with $\smash{\breve{\alpha}_{m_1 \dots m_{q-1}} = \breve{a}_{m_1 \dots m_{q-1}}/l_s^{q-1}}$. After the rescaling $\smash{\zeta_{\R\R^q} = (g_s \ab \breve{a}_{q-1} \ab / \sqrt{2 \smash{\kappa_d^2}}) \, \theta_{\R\R^q}}$, one finds the action
    \begin{equation} \label{RR-axion action}
        S[\zeta_{\R\R^q}] = \int_{\mathrm{X}_{1,d-1}} \de^{1,d-1} x \, \sqrt{- \tilde{g}_{1,d-1}} \; \biggl[ - \dfrac{1}{2} \, \e^{\sqrt{d-2} \, \kappa_d \tdelta - \frac{[2 (q-1) - (10-d)]}{\sqrt{10-d}} \, \kappa_d \tsigma} \, (\der \zeta_{\R\R^q})^2 \biggr].
    \end{equation}
    If $d=4$ and $q=1, 5$, one recovers the results for the type-IIB axio-dilaton and Kähler-modulus supermultiplets.
    \item In a similar way, one can compute the kinetic term associated to the NSNS-axion via $\smash{B_2 = \theta_{\NS\NS} \, \breve{\alpha}_{2} l_s^{2}}$ in type-II and heterotic superstring theories, in the same notation as above. In view of the string-frame action
    \begin{align*}
        S[H_3] & = \dfrac{1}{2 \kappa_{10}^2} \int_{\mathrm{X}_{1,9}} \biggl[ - \dfrac{1}{2} \, \e^{-2 \Phi} \, H_3 \wedge \star_{1,9} H_3 \biggr],
    \end{align*}
    one can adapt the result above immediately. Besides the flux-induced potential in ref.~\cite[eq.~(B.10)]{Shiu:2023fhb}, after the rescaling $\smash{\zeta_{\NS\NS} = (\ab \breve{a}_{2} \ab / \sqrt{2 \smash{\kappa_d^2}}) \, \theta_{\NS\NS}}$, one gets
    \begin{equation} \label{NSNS-axion action}
        S[\zeta_{\NS\NS}] = \int_{\mathrm{X}_{1,d-1}} \de^{1,d-1} x \, \sqrt{- \tilde{g}_{1,d-1}} \; \biggl[ - \dfrac{1}{2} \, \e^{- \frac{4 \, \kappa_d \tsigma}{\sqrt{10-d}}} \, (\der \zeta_{\NS\NS})^2 \biggr].
    \end{equation}
    \item In a similar way, one can compute the kinetic term associated with the heterotic YM-axion via $\smash{A_{\mathrm{YM}} = \theta_{\mathrm{YM}} \, \breve{\alpha}_{1} l_s^2}$, in the same notation as above. In view of the string-frame action
    \begin{align*}
        S[F_{\mathrm{YM}}] & = \int_{\mathrm{X}_{1,9}} \biggl[ - \dfrac{1}{2 g_{10}^2} \, \e^{-2 \Phi} \, F_{\mathrm{YM}} \wedge \star_{1,9} F_{\mathrm{YM}} \biggr],
    \end{align*}
    where $\smash{g_{10}^2 = 16 \pi^2 l_s^2 \kappa_{10}^2}$, one can adapt the result above immediately. Besides the flux-induced potential in ref.~\cite[eq.~(B.12)]{Shiu:2023fhb}, after the rescaling $\smash{\zeta_{\mathrm{YM}} = \bigl[ (\ab \breve{a}_{1} \ab / (4 \pi \kappa_d) \bigr] \, \theta_{\mathrm{YM}}}$, one gets
    \begin{equation} \label{YM-axion action}
        S[\zeta_{\mathrm{YM}}] = \int_{\mathrm{X}_{1,d-1}} \de^{1,d-1} x \, \sqrt{- \tilde{g}_{1,d-1}} \; \biggl[ - \dfrac{1}{2} \, \e^{- \frac{2 \, \kappa_d \tsigma}{\sqrt{10-d}}} \, (\der \zeta_{\mathrm{YM}})^2 \biggr].
    \end{equation}
\end{itemize}
One can read off the $\lambda_{ra}$-couplings for RR-, NSNS- and YM-axions immediately from eqs.~(\ref{RR-axion action}, \ref{NSNS-axion action}, \ref{YM-axion action}). These can be compared with the scalar potentials generated by RR-, NSNS-, and YM-fluxes in ref.~\cite[eqs.~(B.10, B.12, B.14)]{Shiu:2023fhb}.


\bibliographystyle{apsrev4-1}
\bibliography{report.bib}

\end{document}